\documentclass{article}
\usepackage{ijcai26}

\usepackage{times}
\usepackage{soul}
\usepackage{url}
\usepackage[hidelinks]{hyperref}
\usepackage[utf8]{inputenc}
\usepackage[small]{caption}
\usepackage{graphicx}
\usepackage{amsmath}
\usepackage{amssymb}
\usepackage{amsfonts}
\usepackage{amsthm}
\usepackage{booktabs}
\usepackage{algorithm}
\usepackage{algorithmic}
\usepackage[switch]{lineno}
\usepackage{enumitem}
\usepackage{bm}
\usepackage[table]{xcolor}
\usepackage{subcaption}
\usepackage{multirow} 
\usepackage{mdframed}
\newtheorem{assumption}{Assumption}
\newtheorem{prop}{Proposition}
\newtheorem{lemma}{Lemma}

\newtheorem{theorem}{Theorem}

\title{Uncovering Residual Factors in Financial Time Series via PCA and MTP$_2$-constrained Gaussian Graphical Models}

\author{
Koshi Watanabe$^1$
\and
Ryota Ozaki$^2$\and
Kentaro Imajo$^2$\And
Masanori Hirano$^2$\\
\affiliations
$^1$Hokkaido University\\
$^2$Preferred Networks Inc.
\emails
koshi.watanabe.504@gmail.com,
\{ryota55ozaki, imos\}@preferred.jp,  research@mhirano.jp
}

\begin{document}

\maketitle

\begin{abstract}
    Financial time series are commonly decomposed into market factors, which capture shared price movements across assets, and residual factors, which reflect asset-specific deviations.
    To hedge the market-wide risks, such as the COVID-19 shock, trading strategies that exploit residual factors have been shown to be effective.
    However, financial time series often exhibit near-singular eigenstructures, which hinder the stable and accurate estimation of residual factors.
    This paper proposes a method for extracting residual factors from financial time series that hierarchically applies principal component analysis (PCA) and Gaussian graphical model (GGM).
    Our hierarchical approach balances stable estimation with elimination of factors that PCA alone cannot fully remove, enabling efficient extraction of residual factors.  
    We use multivariate totally positive of order 2 (MTP$_2$)-constrained GGM to capture the predominance of positive correlations in financial data.
    Our analysis proves that the resulting residual factors exhibit stronger orthogonality than those obtained with PCA alone.
    Across multiple experiments with varying test periods and training set lengths, the proposed method consistently achieved superior orthogonality of the residual factors. 
    Backtests on the S\&P~500 and TOPIX~500 constituents further indicate improved trading performance, including higher Sharpe ratios.
\end{abstract}

\section{Introduction}\label{Sec:Introduction}
Financial markets exhibit complex price dynamics driven by both market-wide and asset-specific risks.
Figure~\ref{Fig:Return-Series} depicts the cumulative stock returns of Google and Apple. 
The two series move together, rising and falling at similar times, yet each exhibit distinct asset-specific deviations.
Classical factor models, such as CAPM \cite{sharpe1964capital} and the Fama--French model \cite{fama1992cross}, formalize this by decomposing returns into \emph{common factors} that affect many assets and \emph{residual factors} unique to each asset. 
\par
Identifying residual factors has attracted considerable attention in both research and practice. 
For example, recent studies show that residual factors enable construction of portfolios hedged against common market factors \cite{blitz2011residual,imajo2021deep}. 
These findings indicate that accurate extraction of residual factors have substantial potential for trading and market analysis.
\par
Because residual factors are independent across assets, they can be identified using multivariate analysis.
Two widely used approaches for recovering independent factors are \textbf{whitening} \cite{kessy2018optimal} and \textbf{independent component analysis (ICA)} \cite{hyvarinen2001independent,hyvarinen2023nonlinear}. Whitening equalizes eigenvalues of the data covariance, and ICA typically follows by rotating the whitened data to achieve statistical independence.
These methods have proven effective in domains such as computer vision and speech recognition \cite{coates2011analysis,huang2018decorrelated,choi2005blind}.
However, financial data streams often yield nearly singular covariance matrices, and equalizing small eigenvalues linked to asset-specific fluctuations can cause numerical instability.
\begin{figure}[t]
    \centering
    \includegraphics[width=80mm]{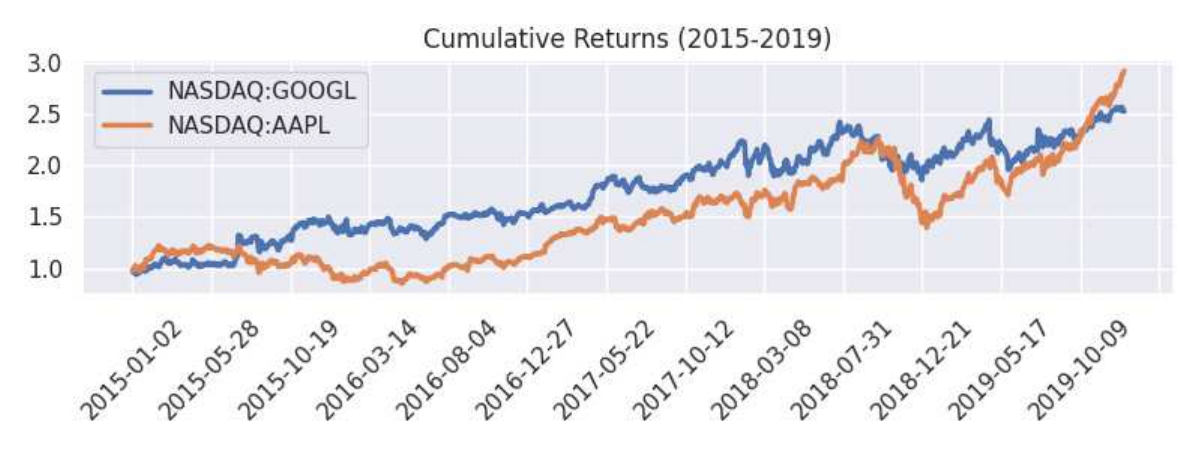}
    \caption{Cumulative stock returns for Google (NASDAQ: GOOGL) and Apple (NASDAQ: AAPL), 2015--2019.}
    \label{Fig:Return-Series}
\end{figure}
\par
In contrast, researchers have proposed methods that extract residual factors by applying \textbf{principal component analysis (PCA)} to remove dominant market components, grounded in arbitrage pricing theory \cite{ross1976arbitrage,chamberlain1983arbitrage}.
These methods are robust to near-singular covariance structures because they remove the eigenvectors associated with the large eigenvalues.
PCA-based methods have been employed in numerous studies \cite{onatski2010determining,ahn2013eigenvalue,bai2016econometric} and are now recognized as a fundamental approach for extracting market and residual factors.
\par
However, PCA-based methods can exhibit reduced accuracy when estimating weak factors \cite{chudik2011weak,onatski2012asymptotics,bai2016econometric}.
Studies applying random matrix theory to financial time series \cite{laloux1999noise} find that within-sector correlations are substantially weaker than market-wide correlations \cite{plerou2002random,pan2007collective}.
PCA-based methods often rely on information criteria to determine the number of components, by assessing the relative explanatory power of each factor.
As a result, factors associated with relatively weak intra-sector correlations may be underestimated, leading to deterioration in identification of residual factors.
\par
To address this issue, we adopt Gaussian graphical models (GGMs) as an effective approach \cite{yuan2007model,friedman2008sparse}.
GGMs are widely used to infer sparse patterns of partial correlations and are well suited for modeling residual dependencies in return data after controlling for dominant factors using PCA.
Furthermore, previous studies document positive cross-asset correlations \cite{sharpe1964capital,chamberlain1983arbitrage,fama1992cross}, which provide a useful inductive bias when inferring the remaining correlation structure.
\par
This paper presents a novel method for identifying residual factors in financial time series: it first applies PCA, then eliminates factors estimated by GGM constrained by multivariate total positivity of order 2 (MTP$_2$) \cite{lauritzen2019maximum,ying2023adaptive}.
The MTP$_2$ constraint has been recently studied for its tractability and wide applications \cite{slawski2015estimation,lauritzen2019maximum,agrawal2022covariance,ying2023adaptive} and restricts the class of covariance between assets such that all partial correlations are nonnegative.
This restriction is consistent with empirical characteristics of financial data streams and is validated in covariance estimation \cite{agrawal2022covariance}.
\par
Our analysis is demonstrated through a theoretical comparison with the PCA method and experiments using historical data from two stock markets, S\&P~500 and TOPIX~500 constituents. 
First, under the MTP$_2$ assumption, we show that residual factors obtained by the proposed method exhibit stronger orthogonality than those obtained via PCA. 
This finding is independent of the method used to determine the number of factors \cite{bai2002determining,onatski2010determining,ahn2013eigenvalue}, suggesting that the approach can extend to other PCA-based factor estimation studies.
In our experiments, we perform backtests on more than ten years of historical data.
We compare the orthogonality of the residual returns obtained by the proposed method against those by the PCA or the other whitening methods, confirming that the theoretical result also holds empirically. 
We further conduct contrarian trading using the residual factors, showing that the proposed method achieves superior performance, including higher Sharpe ratios. 
All experiments are conducted comprehensively by rolling through multiple evaluation periods.
\par
The \textbf{main contributions} of this study are as follows:
\begin{itemize}[itemsep=0pt, parsep=1pt]
\item We propose a novel method for extracting residual factors from financial time series, enabling accurate and robust estimation.
\item Under assumptions appropriate for financial time series, the method produces residual factors that are more orthogonal across assets than those from PCA.
\item Experiments using historical data on S\&P~500 and TOPIX~500 constituents demonstrate greater orthogonality and stronger performance of reversal strategies across multiple time periods.
\end{itemize}

\section{Preliminaries}\label{Sec:Preliminaries}
\subsection{Problem Setting}\label{SubSec:Setting}
Let $p_{n,t}$ denote the price of asset $n$ observed at time $t$.  
The corresponding (arithmetic) return is defined as $x_{n,t} = \frac{p_{n,t+1}}{p_{n,t}} - 1$. 
Let $\mathbf{x}_{n} = [x_{n,1}, x_{n,2}, \ldots, x_{n,T}]^{\top} \in \mathbb{R}^{T}$ denote the vector of returns for asset $n$ over $T$ consecutive periods.
This study considers the return matrix $\mathbf{X} = [\mathbf{x}_1,\, \mathbf{x}_2,\, \ldots,\, \mathbf{x}_N]^{\top} \in \mathbb{R}^{N \times T}$, which comprises return sequences of $N$ assets, and addresses the problem of extracting residual factors from 
the cross-asset return vector at time $t$, denoted by $\mathbf{x}_{:,t} \in \mathbb{R}^{N}$. 
\par
In datasets with irregular variation, including financial time series, nonlinear transformations are prone to overfitting and can lead to unstable out-of-sample performance.
Therefore, this study aims to estimate a transformation matrix $\mathbf{W} \in \mathbb{R}^{N \times N}$ that derives a residual matrix $\mathbf{R}=\mathbf{W}\mathbf{X}$.  
Our objective is to minimize the absolute Pearson correlation between assets $i$ and $j$, given by $\frac{|\mathbf{r}_i^{\top}\mathbf{r}_j|}{\|\mathbf{r}_i\|_2\|\mathbf{r}_j\|_2}$.

\subsection{Related Works}\label{SubSec:Related-Works}
The most straightforward approach is \emph{whitening} \cite{hyvarinen2001independent,kessy2018optimal,hyvarinen2023nonlinear}, 
which entails estimating a transformation matrix $\mathbf{W}$ such that $\mathbf{R}\mathbf{R}^{\mathsf{T}} =\mathbf{I}$ in the training data.
While this approach shows promise, the computation of the inverse of the eigenvalues to equalize variances 
is a primary source of numerical instability in financial time series with
substantial variation in the eigenvalue spectrum.
One potential solution introduces shrinkage estimation \cite{ledoit2004well,chen2010shrinkage}, a framework for stabilizing eigenvalues. 
While this approach improves numerical stability, it may simultaneously degrade the estimation accuracy of residual factors.
\par
Based on PCA, residual factors are obtained by removing contribution of the largest principal components from the return series \cite{bai2002determining,fan2013large,imajo2021deep}. 
This approach typically uses an information criterion to determine the number of statistically significant factors, with remaining components treated as residual factors. 
Several extensions have been proposed \cite{bai2016econometric}, such as introducing new information criteria \cite{onatski2010determining,ahn2013eigenvalue} or modeling time variation \cite{ahn2013panel,bai2021dynamic}. 
Although our proposed method builds upon the most fundamental approach \cite{bai2002determining}, the analysis extends naturally to these alternative settings.

\section{PCA with an Information Criterion} \label{Sec:PCA-IC}
To illustrate PCA-based factor removal, we first define the compact singular value decomposition (SVD) of the return matrix as $\mathbf{X}=\mathbf{U}\bm{\Sigma}\mathbf{V}^{\top}$, where $\mathbf{U}\in\mathbb{R}^{N \times M}$ and $\mathbf{V}\in\mathbb{R}^{T \times M}$ are matrices satisfying the orthonormal conditions $\mathbf{U}^{\top}\mathbf{U}=\mathbf{V}^{\top}\mathbf{V}=\mathbf{I}$, with $M=\min(N, T)$.  
The matrix $\bm{\Sigma}=\mathrm{diag}(\sigma_{1},\sigma_{2},\ldots,\sigma_{M})$ is diagonal, and its entries correspond to singular values arranged in descending order, i.e., $\sigma_{1}>\sigma_{2}>\cdots >\sigma_{M}>0$.
An equivalent representation is $\mathbf{X}=\sum_{i=1}^{M}\sigma_{i}\mathbf{u}_{:,i}\mathbf{v}_{:,i}^{\top}$, where $\mathbf{u}_{:,i} \in \mathbb{R}^{N}$ and $\mathbf{v}_{:,i} \in \mathbb{R}^{T}$ denote the $i$-th column vectors of $\mathbf{U}$ and $\mathbf{V}$, respectively. 
The extraction of the factor corresponding to the $k$-th singular value can be formulated using a linear projection operator.  
Specifically, since the orthonormal condition 
$\mathbf{u}_{:,i}^{\top}\mathbf{u}_{:,j} =
\begin{cases}
    1 & (i = j), \\
    0 & (i \neq j)
\end{cases}$ is satisfied, premultiplying the return matrix by $\mathbf{u}_{:,k}\mathbf{u}_{:,k}^{\top} \in \mathbb{R}^{N \times N}$ yields  
\begin{align}
    &\mathbf{u}_{:,k}\mathbf{u}_{:,k}^{\top}\mathbf{X}   
    = \sum_{i=1}^{M} \sigma_{i} \mathbf{u}_{:,k} \mathbf{u}_{:,k}^{\top} \mathbf{u}_{:,i} \mathbf{v}_{:,i}^{\top}   
    = \sigma_{k} \mathbf{u}_{:,k} \mathbf{v}_{:,k}^{\top}.
    \label{Eq:Factor-Extraction}
\end{align}
\par
Building on the above equation, this step extracts components corresponding to the first $k$ singular values.  
By defining $\mathbf{U}_k = [\mathbf{u}_{:,1},\, \mathbf{u}_{:,2},\, \ldots,\, \mathbf{u}_{:,k}] \in \mathbb{R}^{N \times k}$,
which contains the first $k$ column vectors of the left singular vector matrix,  
we construct a linear projection that removes the corresponding principal factors as $\mathbf{W}_{\text{PCA}} = \mathbf{I} - \mathbf{U}_k \mathbf{U}_k^\top$.
\par
To determine the number of components $k$, the method employs an information criterion.  
Specifically, $k$ is selected to maximize  
\begin{align}
    IC_{k} = \log \sum_{i=k+1}^{M} \sigma_{i}^{2} + k \frac{\log(M)}{M}.  \label{Eq:Information-Criteria}
\end{align}
For purely idiosyncratic fluctuations associated with small singular values, we have $\sum \sigma_{i}^{2} \approx 0$, implying that the first term in Eq. \eqref{Eq:Information-Criteria} takes an extremely small value.  
Consequently, maximizing Eq.~\eqref{Eq:Information-Criteria} enables selection of singular vectors corresponding to residual factors, while mitigating numerical instability and reducing overfitting risk caused by small eigenvalues.
\par
\begin{figure}[t]
  \centering
  \begin{subfigure}[t]{0.21\textwidth}
    \includegraphics[width=\linewidth]{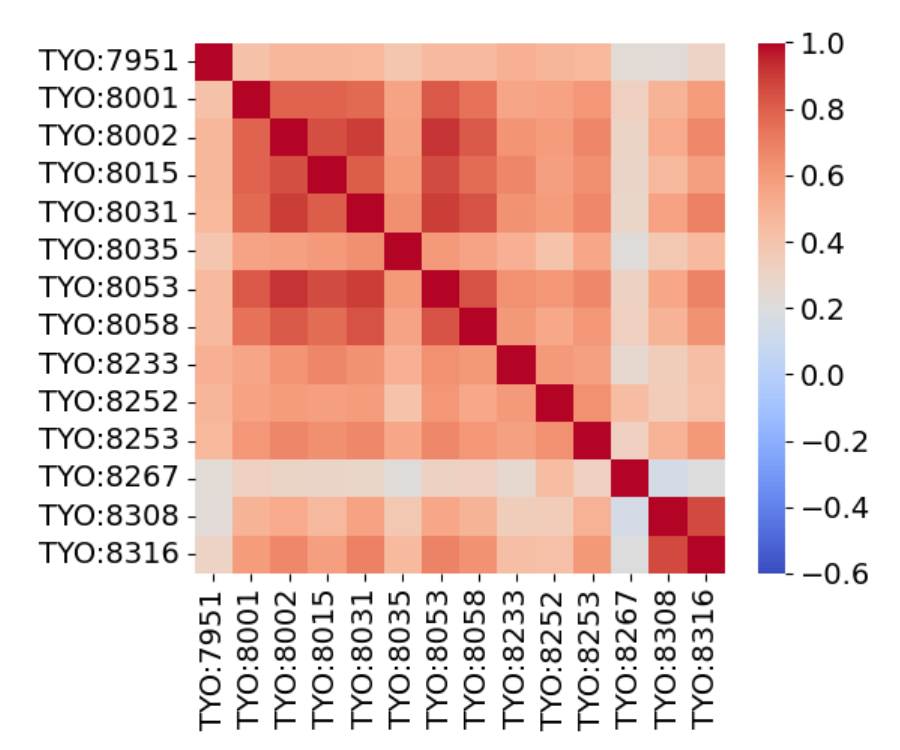}
    \caption{Raw returns}
  \end{subfigure}
  \begin{subfigure}[t]{0.21\textwidth}
    \centering
    \includegraphics[width=\linewidth]{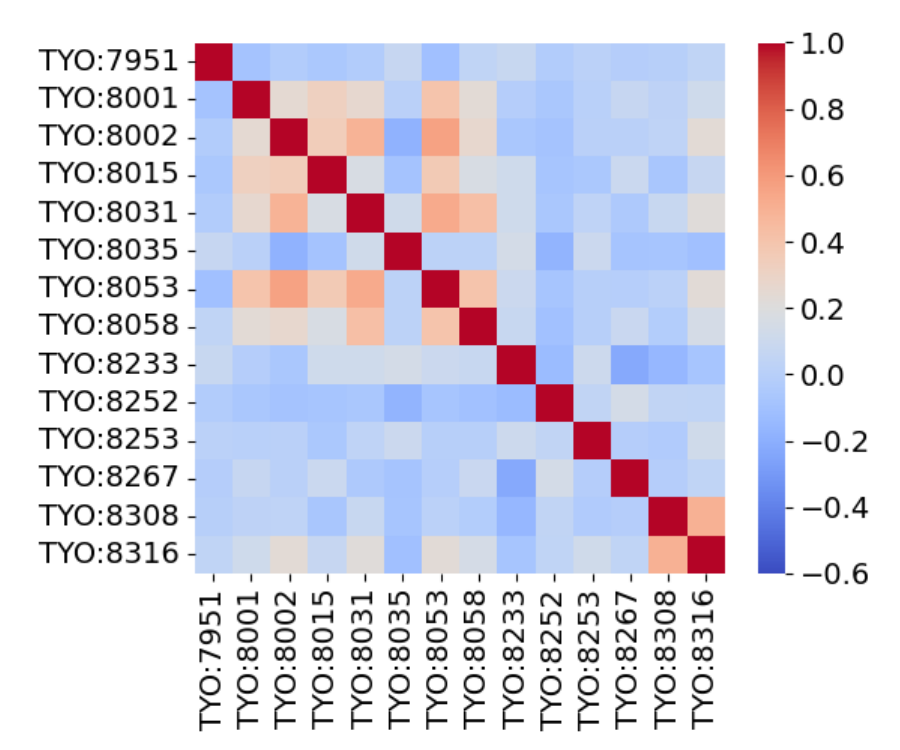}
    \caption{PCA residual factors}
  \end{subfigure}
  \caption{Correlation matrices between assets of TOPIX~500 (2024).}
  \label{Fig:PCA-IC-Results}
\end{figure}
However, this method has the limitation that it cannot remove relatively weak factors.  
Figure~\ref{Fig:PCA-IC-Results} depicts correlation matrices of (a) raw returns and (b) the PCA residual factors obtained by maximizing Eq.~\eqref{Eq:Information-Criteria} for TOPIX~500 constituents.
These results indicate that, while the PCA residual factors successfully eliminate market-wide factors, certain partial correlations remain (e.g., among \texttt{TYO:8001}, \texttt{TYO:8002}, \texttt{TYO:8015}, and \texttt{TYO:8031}) and these assets belong to the same industry sector (e.g., \texttt{TYO:80**} denotes the trading company sector).
This finding illustrates that the information criterion-based approach does not eliminate relatively weak factors, such as sector-specific factors, and there remains room for improvement in identifying fully independent residual factors.

\section{Proposed Method: MTP$_2$ GGM-based Factor Removal} \label{Section:MTP2-GGM-Factor-Removal}
\subsection{Optimization Problem} \label{SubSec:Optimization}
The proposed method employs a Gaussian graphical model (GGM) to remove common factors that are not removed by PCA.  
Let $\mathbf{Z} = \mathbf{W}_{\mathrm{PCA}} \mathbf{X} \in \mathbb{R}^{N \times T}$ denote the return matrix after PCA-based factor removal.
GGM assumes that the cross-asset return vectors $\mathbf{z}_{:,t}$ are generated from a normal distribution $\mathcal{N}(\mathbf{0}, \bm{\Lambda}^{-1})$, and our objective is to estimate the precision matrix $\bm{\Lambda}$. 
Additional constraints can be imposed as inductive bias for GGMs.
Since financial time series frequently exhibit positive correlations, as shown in Figure~\ref{Fig:PCA-IC-Results}, the method imposes the constraint that the precision matrix $\bm{\Lambda}$ belongs to the set of $N \times N$ symmetric \textit{M}-matrices \cite{plemmons1977m} $\mathcal{M}^N$:
\begin{equation}
    \begin{aligned}
        \bm{\Lambda} 
        &= \max_{\bm{\Lambda}}
           \left\{ \log|\bm{\Lambda}| - \frac{1}{T} \mathrm{tr}\left( \bm{\Lambda} \mathbf{Z} \mathbf{H} \mathbf{Z}^{\top} \right) \right\}, \\
        & \quad \quad \quad \quad \text{s.t.} \quad \bm{\Lambda} \in \mathcal{M}^{N}.
    \end{aligned}
    \label{Eq:MTP2-Problem}
\end{equation}
where $\mathbf{H} = \mathbf{I} - \frac{1}{T} \mathbf{1} \mathbf{1}^{\top}$ is the centering matrix, and $\mathcal{M}^{N}$ denotes the set of symmetric \textit{M}-matrices, defined as positive semi-definite matrices with non-positive off-diagonal entries.
\par
We examine implications of the \textit{M}-matrix constraint.  
Non-positivity of off-diagonal elements of $\bm{\Lambda}$ ensures that the \emph{partial correlation coefficient} $\rho_{i,j} = -\frac{\Lambda_{i,j}}{\sqrt{\Lambda_{i,i} \Lambda_{j,j}}}$ between any two assets is positive.  
In a hypothetical scenario where prices of certain stocks are known at a specific time, we assume that \emph{the positively correlated movements of other stocks are predictable, whereas movements lacking positive correlation are difficult to anticipate}.
This captures the common market hypothesis \cite{fama1992cross} that asset prices move in tandem within specific sectors or groups of stocks. 
Consequently, imposing the \textit{M}-matrix constraint facilitates extraction of common factors by assuming rational dynamics among assets.

\subsection{Optimization Algorithm}\label{SubSec:Algorithm}
The solution to Eq.~\eqref{Eq:MTP2-Problem} cannot be obtained in closed form and is typically solved using gradient-based methods~\cite{ying2023adaptive}. 
We employ the projected gradient method to solve the problem.  
This algorithm updates $\bm{\Lambda}$ iteratively in the gradient direction then projects it back onto the feasible region $\mathcal{M}^{N}$ (\textbf{Algorithm~\ref{Alg:Projected-GA}}).
The projection operator $\Pi_{\mathcal{M}^{N}}(\cdot)$ is the shortest-distance projection that minimizes the Frobenius norm between the updated matrix and the feasible set:
\begin{equation}
    \Pi_{\mathcal{M}^{N}}(\bm{\Lambda})=\arg \min_{\mathbf{Y}\in\mathcal{M}^{N}}||\bm{\Lambda}-\mathbf{Y}||_{\mathrm{F}}^{2}. \label{Eq:Prox-$M$-matrix}
\end{equation}
Note that the projection in Eq.~\eqref{Eq:Prox-$M$-matrix} is equivalent to the proximal operator when the indicator function 
${1}_{\mathcal{M}^{N}}=\begin{cases}
    0&(\mathbf{\Lambda} \in \mathcal{M}^{N}) \\ \infty & (\mathrm{otherwise})
\end{cases}$ 
is used as a penalty term in problem~\eqref{Eq:MTP2-Problem}. 
This guarantees convergence of \textbf{Algorithm~\ref{Alg:Projected-GA}} to at least a saddle point.
\par
However, Eq.~\eqref{Eq:Prox-$M$-matrix} also cannot be solved analytically in closed form.  
To address this issue, we employ Dykstra's projection method~\cite{boyle1986method}, which computes the projection onto the intersection of multiple convex sets (\textbf{Algorithm \ref{Alg:Prox-$M$-matrix}}).
The set $\mathcal{M}^{N}$ can be expressed as the intersection of the $N$-dimensional cone of positive semi-definite matrices $\mathcal{S}_{+}^{N}$ and the set of $N$-dimensional matrices with non-positive off-diagonal elements $\mathcal{Z}^{N}$, as
\begin{equation}
    \mathcal{M}^{N} = \mathcal{S}_{+}^{N} \cap \mathcal{Z}^{N}.
    \label{Eq:Set-$M$-matrix}
\end{equation}
Projections onto each convex set are expressed as follows:
\begin{align}
    \Pi_{\mathcal{S}_{+}^{N}}(\bm{\Lambda})&=\arg \min_{\mathbf{Y}\in\mathcal{S}_{+}^{N}}||\bm{\Lambda}-\mathbf{Y}||_{\mathrm{F}}^{2}, \label{Eq:Prox-PD-Cone} \\
    \Pi_{\mathcal{Z}^{N}}(\bm{\Lambda})&=\arg \min_{\mathbf{Y}\in\mathcal{Z}^{N}}||\bm{\Lambda}-\mathbf{Y}||_{\mathrm{F}}^{2}. \label{Eq:Prox-Z-Matrix}
\end{align}
While computing $\Pi_{\mathcal{M}^{N}}(\cdot)$ is computationally challenging, its constituent projections $\Pi_{\mathcal{S}_{++}^{N}}(\cdot)$ and $\Pi_{\mathcal{Z}^{N}}(\cdot)$ admit closed-form expressions.
\begin{theorem}
    Let $\bm{\Lambda}$ be an $N \times N$ square matrix and let its eigenvalue decomposition be $\bm{\Lambda} = \mathbf{Q}\mathbf{\Omega}\mathbf{Q}^{\top}$.  
    Then, the projections $\Pi_{\mathcal{S}_{++}^{N}}(\cdot)$ and $\Pi_{\mathcal{Z}^{N}}(\cdot)$ are given by
    \begin{align}
        \Pi_{\mathcal{S}_{++}^{N}}(\bm{\Lambda}) &= \mathbf{Q}\mathbf{\Omega}_{+}\mathbf{Q}^{\top}, \label{Eq:Closed-Form-Prox-PD-Cone} \\
        \Pi_{\mathcal{Z}^{N}}(\bm{\Lambda}) &=
        \begin{cases}
            \min(\Lambda_{i,j},0) & (i \neq j), \\
            \Lambda_{i,j} & (i = j),
        \end{cases}
        \label{Eq:Closed-Form-Prox-Z-Matrix}
    \end{align}
    where $\mathbf{\Omega}_{+} = \mathrm{diag}(\max(\omega_{1},0),\ldots,\max(\omega_{N},0))$ denotes the diagonal matrix containing positive eigenvalues.
    \label{Thrm:S-Z-Projection}
\end{theorem}
\begin{proof}
    See Appendix \ref{Apx:Proof-S-Z-Projection}.
\end{proof}
We solve Eq.~\eqref{Eq:MTP2-Problem} by iteratively applying these algorithms.
\begin{algorithm}[t]
    \caption{Projected gradient ascent for Problem \eqref{Eq:MTP2-Problem}.}
    \begin{algorithmic}[1]
        \STATE $\bm{\Lambda}^{(0)} = \Pi_{\mathcal{M}^{N}}(\mathbf{Z}\mathbf{H}\mathbf{Z}^{\top})$, learning rate $\alpha$, number of epochs \texttt{O\_ITER}, $k \gets 0$
        \WHILE{$k < \texttt{O\_ITER}$}
            \STATE $\bm{\Lambda}'  \gets \bm{\Lambda}^{(k)} + \alpha \nabla f(\bm{\Lambda}^{(k)})$
            \STATE $\bm{\Lambda}^{(k+1)} \gets \Pi_{\mathcal{M}^{N}}(\bm{\Lambda}')$
            \STATE $k \gets k + 1$
        \ENDWHILE
        \RETURN $\bm{\Lambda}^{(k)}$
    \end{algorithmic}
    \label{Alg:Projected-GA}
\end{algorithm}
\begin{algorithm}[t]
    \caption{Dykstra's projection method for $\Pi_{\mathcal{M}^{N}}(\tilde{\bm{\Lambda}})$}
    \begin{algorithmic}[1]
        \STATE $\bm{\Lambda}'^{(0)} \gets \bm{\Lambda}'$, $\mathbf{E}_{S}^{(0)} \gets \mathbf{O}$, $\mathbf{E}_{Z}^{(0)} \gets \mathbf{O}$, number of iterations \texttt{I\_ITER}, $k \gets 0$
        \WHILE{$k < \texttt{I\_ITER}$}
            \STATE $\mathbf{A}^{(k)} \gets \Pi_{\mathcal{S}_{+}^{N}}\left( \bm{\Lambda}'^{(k)} + \mathbf{E}^{(k)}_{S} \right)$ 
            \STATE $\mathbf{E}^{(k+1)}_{S} \gets \bm{\Lambda}'^{(k)} + \mathbf{E}^{(k)}_{S} - \mathbf{A}^{(k)}$
            \STATE $\bm{\Lambda}'^{(k+1)} \gets \Pi_{\mathcal{Z}^{N}}\left( \mathbf{A}^{(k)} + \mathbf{E}_{Z}^{(k)} \right)$
            \STATE $\mathbf{E}_{Z}^{(k+1)} \gets \mathbf{A}^{(k)} + \mathbf{E}_{Z}^{(k)} - \bm{\Lambda}'^{(k+1)}$
            \STATE $k \gets k + 1$
        \ENDWHILE
        \RETURN $\bm{\Lambda}'^{(k)}$
    \end{algorithmic}
    \label{Alg:Prox-$M$-matrix}
\end{algorithm}
\subsection{Projection for Factor Removal}\label{SubSec:MTP2-Factor-Removal}
We obtain the optimal $\bm{\Lambda}$ using \textbf{Algorithm \ref{Alg:Projected-GA}} and \textbf{Algorithm \ref{Alg:Prox-$M$-matrix}}.
The method then transforms the precision matrix back into information in the return space.
We employ the conditional mean of the return to achieve this, given by
\begin{align}
\mathbb{E}[z_{i,t} \mid \mathbf{z}_{-i,t}] &= -\sum_{j \neq i} \frac{\Lambda_{i,j}}{\Lambda_{i,i}} z_{j,t}, \label{Eq:GGM-Factor}
\end{align}
where $\mathbf{z}_{-i,t}$ denotes the vector obtained by deleting the $i$-th element from $\mathbf{z}_{:,t}$.
Equation \eqref{Eq:GGM-Factor} is equivalent to predicting the factor of asset $i$ at time $t$ based on factors of all other assets.
The method treats this predicted return as the common factor associated with asset $i$, since it can be estimated from other correlated returns. 
Therefore, the remaining residual factors are given as follows:
\begin{align}
    z_{i,t}-\mathbb{E}[z_{i,t}|\mathbf{z}_{-i,t}]&=\sum_{j=1}^{N}\frac{\Lambda_{i,j}}{\Lambda_{i,i}}x_{j,t}. \label{Eq:GGM-Residuals}
\end{align}
Let the diagonal elements of $\bm{\Lambda}$ as $\mathbf{D}$.
Then, the transformation by Eq. \eqref{Eq:GGM-Residuals} can be written as $\mathbf{W}_{\mathrm{GGM}}=\mathbf{D}^{-1}\bm{\Lambda}$.
\par
From the above, the method derives residual factors by hierarchically adopting PCA and MTP$_2$-based removal, and the resulting linear projector is given by:
\begin{align}
    \mathbf{W}&=\mathbf{W}_{\mathrm{GGM}}\mathbf{W}_{\mathrm{PCA}} \notag\\
    &=\mathbf{D}^{-1}\bm{\Lambda}(\mathbf{I}-\mathbf{U}_{k}\mathbf{U}^{\top}_{k}).
\end{align}
The residual factors of our method is also given as:
\begin{align}
    \mathbf{R}=\mathbf{W}_{\mathrm{GGM}}\mathbf{W}_{\mathrm{PCA}}\mathbf{X}=\mathbf{W}_{\mathrm{GGM}}\mathbf{Z}
\end{align}

\begin{table*}[t]
    \centering
    \caption{The $\ell_{1}$ and $\ell_{2}$ mean of cross-asset correlations with different training/testing periods. Boldface indicates the best, and underlining indicates the second-best.}
    \begin{tabular}{llccccccc}
        \toprule
        Asset & train start (3 years) & 2015/01 & 2016/01 & 2017/01 & 2018/01 & 2019/01 & 2020/01 & 2021/01\\
        & valid start (1 year) & 2018/01 & 2019/01 & 2020/01 & 2021/01 & 2022/01 & 2023/01 & 2024/01 \\
        \midrule
        S\&P~500 & \underline{\textit{$\ell_1$ mean}} & & & & & & &\\
        & ICA & 0.2085 & 0.1995 & 0.2686 & 0.1917 & 0.2001 & 0.2038 & 0.1937 \\
        & Whitening & 0.0763 & 0.0780 & 0.1345 & 0.0771 & 0.0811 & 0.0783 & 0.0776 \\
        & Shr.~Whitening & 0.0709 & 0.0729 & \underline{0.1301} & 0.0731 & 0.0772 & 0.0742 & 0.0727 \\
        & PCA & \underline{0.0670} & \underline{0.0704} & 0.1384 & \underline{0.0678} & \underline{0.0747} & \underline{0.0696} & \underline{0.0683} \\
        & \textbf{PCA+GGM} & \cellcolor{blue!15}\textbf{0.0639} & \cellcolor{blue!15}\textbf{0.0658} & \cellcolor{blue!15}\textbf{0.1270} & \cellcolor{blue!15}\textbf{0.0649} & \cellcolor{blue!15}\textbf{0.0694} & \cellcolor{blue!15}\textbf{0.0663} & \cellcolor{blue!15}\textbf{0.0650} \\
        \midrule
        & \underline{\textit{$\ell_2$ mean}} & & & & & & \\
        & ICA & 0.0727 & 0.0657 & 0.1101 & 0.0591 & 0.0636 & 0.0661 & 0.0616 \\
        & Whitening & 0.0092 & 0.0097 & 0.0295 & 0.0093 & 0.0104 & 0.0097 & 0.0096 \\
        & Shr.~Whitening & 0.0080 & 0.0085 & \underline{0.0278} & 0.0084 & 0.0095 & 0.0088 & 0.0085 \\
        & PCA & \underline{0.0076} & \underline{0.0084} & 0.0312 & \underline{0.0076} & \underline{0.0093} & \underline{0.0082} & \underline{0.0079} \\
        & \textbf{PCA+GGM} & \cellcolor{blue!15}\textbf{0.0068} & \cellcolor{blue!15}\textbf{0.0072} & \cellcolor{blue!15}\textbf{0.0262} & \cellcolor{blue!15}\textbf{0.0068} & \cellcolor{blue!15}\textbf{0.0078} & \cellcolor{blue!15}\textbf{0.0073} & \cellcolor{blue!15}\textbf{0.0071} \\
        \midrule
        TOPIX~500 & \underline{\textit{$\ell_1$ mean}} & & & & & & \\
        & ICA & 0.1874 & 0.1767 & 0.2318 & 0.1912 & 0.1865 & 0.1805 & 0.1829 \\
        & Whitening & 0.0809 & 0.0810 & 0.1177 & 0.0792 & 0.0836 & 0.0824 & 0.0868 \\
        & Shr.~Whitening & 0.0748 & 0.0747 & 0.1126 & 0.0721 & 0.0772 & 0.0752 & 0.0806 \\
        & PCA & \underline{0.0651} & \underline{0.0655} & \underline{0.1102} & \underline{0.0630} & \underline{0.0703} & \underline{0.0654} & \underline{0.0745} \\
        & \textbf{PCA+GGM} & \cellcolor{blue!15}\textbf{0.0645} & \cellcolor{blue!15}\textbf{0.0652} & \cellcolor{blue!15}\textbf{0.1056} & \cellcolor{blue!15}\textbf{0.0620} & \cellcolor{blue!15}\textbf{0.0684} & \cellcolor{blue!15}\textbf{0.0651} & \cellcolor{blue!15}\textbf{0.0731} \\
        \midrule
         & \underline{\textit{$\ell_2$ mean}} & & & & & & & \\
        & ICA & 0.0557 & 0.0501 & 0.0822 & 0.0575 & 0.0554 & 0.0524 & 0.0531 \\
        & Whitening & 0.0103 & 0.0103 & 0.0220 & 0.0098 & 0.0110 & 0.0107 & 0.0119 \\
        & Shr.~Whitening & 0.0088 & 0.0088 & 0.0204 & 0.0082 & 0.0094 & 0.0090 & 0.0103 \\
        & PCA & \underline{0.0070} & \underline{0.0071} & \underline{0.0199} & \underline{0.0066} & \underline{0.0083} & \underline{0.0072} & \underline{0.0092} \\
        & \textbf{PCA+GGM} & \cellcolor{blue!15}\textbf{0.0068} & \cellcolor{blue!15}\textbf{0.0069} & \cellcolor{blue!15}\textbf{0.0180} & \cellcolor{blue!15}\textbf{0.0062} & \cellcolor{blue!15}\textbf{0.0076} & \cellcolor{blue!15}\textbf{0.0071} & \cellcolor{blue!15}\textbf{0.0088} \\
        \bottomrule
    \end{tabular}
    \label{Table:Results-Different-Period}
\end{table*}

\subsection{Theoretical Analysis}
\label{Sec:Analysis}
This subsection examines theoretical properties of residual factors obtained from the proposed method.
We denote by $\mathrm{cov}(\cdot)$ the sample covariance matrix of a given variable, and by $\mathrm{cor}(\cdot)$ its sample correlation matrix.  
The sample covariance between two variables is denoted by $\mathrm{cov}(\cdot,\cdot)$.  
We begin by stating the single assumption required to derive the results.
\begin{assumption}\label{Assum:MTP2}
    We assume that the covariance matrix of $\mathbf{Z}$ satisfies
    \begin{align*}
        \mathrm{cov}(\mathbf{Z}) = \bm{\Lambda}^{-1}.
    \end{align*}
    In other words, the covariance matrix estimated from Problem~\eqref{Eq:MTP2-Problem} is equal to the covariance matrix of $\mathbf{Z}$.
\end{assumption}
Note that stationary points of Problem~\eqref{Eq:MTP2-Problem} without constraints satisfy $\bm{\Lambda}^{-1} = \frac{1}{T}\mathbf{Z} \mathbf{H} \mathbf{Z}^\top = \mathrm{cov}(\mathbf{Z})$, and the MTP$_2$ condition is particularly natural for financial time series. 
Hence, the assumption stated above is well justified.
\par
Using this assumption, we obtain the following property regarding covariance between $\mathbf{R}$ and $\mathbf{Z}$.
\begin{prop}
    The covariance between residual factors $\mathbf{R}$ and PCA residual factors $\mathbf{Z}$ after PCA satisfies
    \begin{align}
        \mathrm{cov}(\mathbf{R}, \mathbf{Z}) &= \mathbf{D}^{-1}, \label{Eq:CovRZ}
    \end{align}
    where $\mathbf{D}$ is a diagonal matrix of $\mathbf{\Lambda}$. 
    \label{Prop:Cov-R}
\end{prop}
\begin{proof}
    See Appendix \ref{Apx:Proof-Prop-Cov-R}.
\end{proof}
\textbf{Proposition \ref{Prop:Cov-R}} implies that residual factors $\mathbf{R}$ are orthogonal to the PCA residual factors.
$\mathbf{R}$ represents the component of $\mathbf{Z}$ that is not explained by other assets and is specific to each asset.
\par
We next present the following inequality regarding correlation between PCA residual factors and those obtained from the proposed method.
\begin{prop}
    Correlation matrices of PCA residual factors $\mathbf{Z}$ and residual factors of proposed method $\mathbf{R}$ satisfy
    \begin{align}
        |\mathrm{cor}(\mathbf{R})| &\leq |\mathrm{cor}(\mathbf{Z})|. \label{Eq:CorR-vs-CorZ}
    \end{align}
    \label{Prop:Cor-R-is-smaller-than-Cor-Z}
\end{prop}
\begin{proof}
See Appendix~\ref{Apx:Proof-Cor-R-is-smaller-than-Cor-Z}.
\end{proof}
\textbf{Proposition \ref{Prop:Cor-R-is-smaller-than-Cor-Z}} shows that introducing GGM-based removal after PCA reduces the cross-asset correlations, thereby making residual factors more orthogonal. 
As discussed earlier, when the residual factors are interpreted as fluctuations driven by asset-specific risks, smaller correlations among residual factors are desirable.
Therefore, this property supports the effectiveness of the proposed method. 

\begin{figure*}[t]
  \centering
  \begin{subfigure}[t]{0.47\textwidth}
    \centering
    \includegraphics[width=\linewidth]{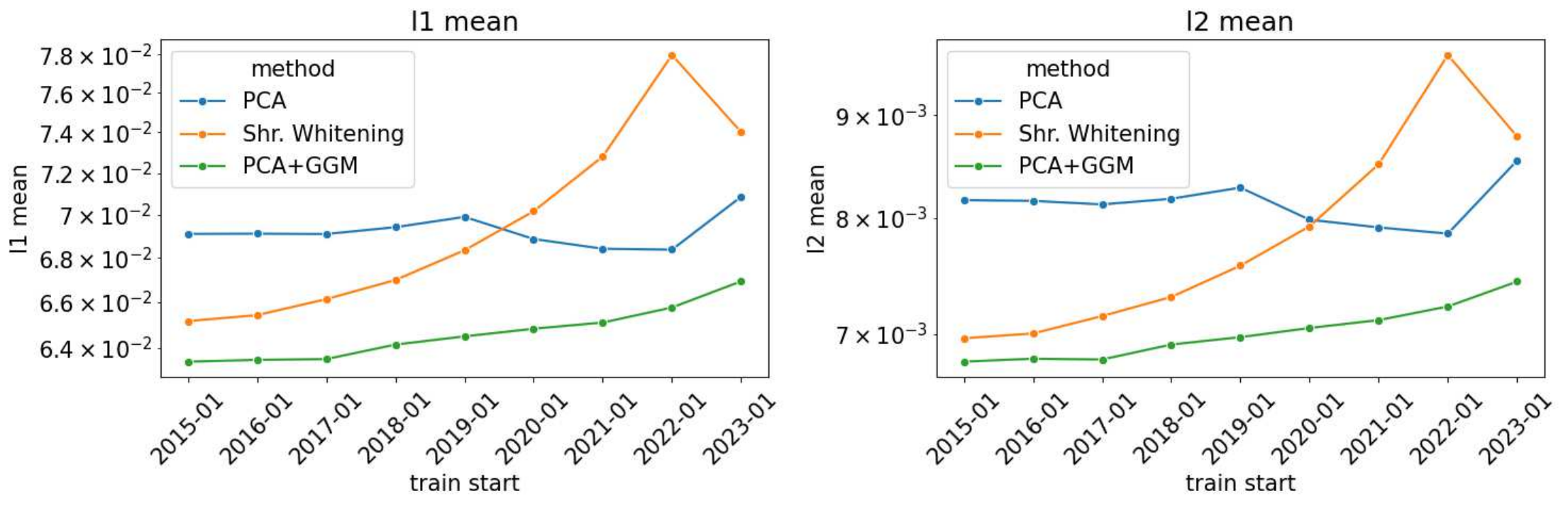}
    \caption{S\&P~500}
  \end{subfigure}
  \hfill
  \begin{subfigure}[t]{0.47\textwidth}
    \centering
    \includegraphics[width=\linewidth]{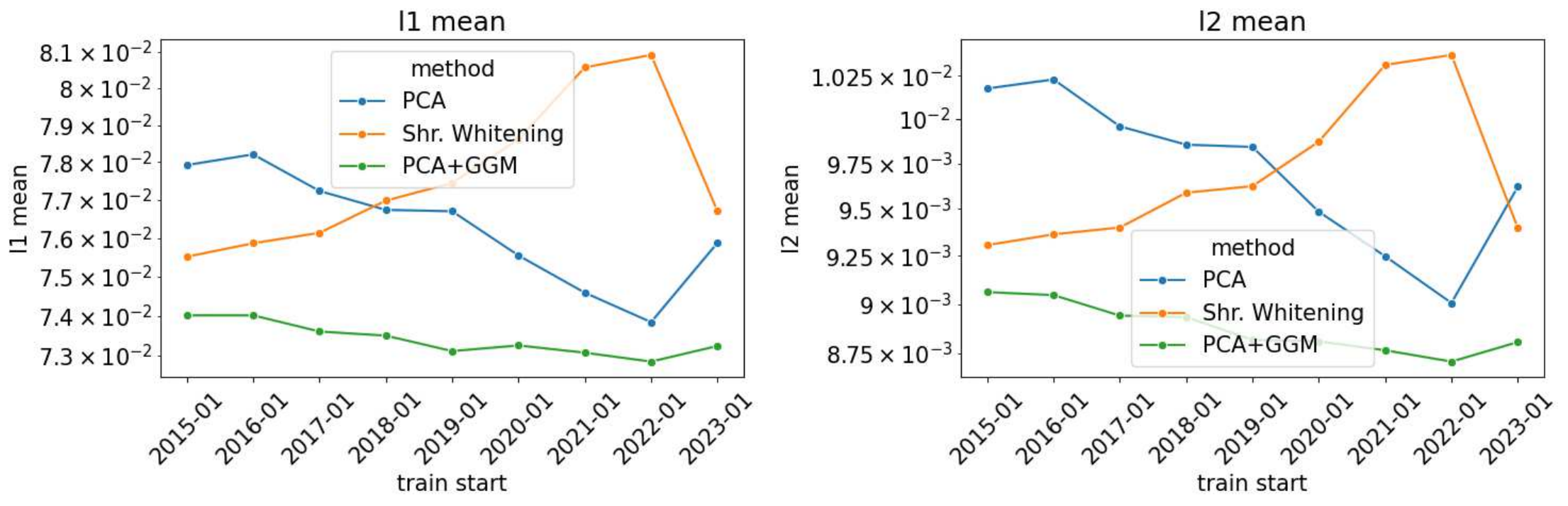}
    \caption{TOPIX~500}
  \end{subfigure}  
  \caption{The $\ell_{1}$ and $\ell_{2}$ means of cross-asset correlations with different training data horizons. The horizontal axis indicates the starting timestamp of the training period. Moving to the right indicates a shorter training period.}
  \label{Fig:Changing-N-Results}
\end{figure*}

\section{Experiments} \label{Sec:Experiment}
\subsection{Experimental Settings} \label{SubSec:Exp-Settings}
We constructed daily return series based on closing prices of constituents of the S\&P~500 (US) and TOPIX~500 (JP).
The resulting return series were split into multiple sets of training and test periods for each experimental setting.
\par
We compared the proposed method with the following methods:
\begin{itemize}
    \item \textbf{ICA}: Independent component analysis (ICA) is a widely used method for orthogonalization that computes a transformation based on kurtosis. In this study, we applied the commonly used Fast ICA algorithm~\cite{hyvarinen2001independent}.
    \item \textbf{Whitening}: Whitening is a method for finding a transformation that orthogonalizes the given data and is considered effective for obtaining mutually independent residual factors. In this study, we applied ZCA whitening~\cite{kessy2018optimal}, which performs orthogonalization in the return space.
    \item \textbf{Shrinkage Whitening (Shr.~Whitening)}: Standard whitening requires taking the inverse of eigenvalues and therefore behaves unstably for singular matrices. To address this, we compared a method that applies the OAS shrinkage estimator~\cite{chen2010shrinkage} to the covariance matrix.
    \item \textbf{PCA}: As a baseline, we used the residual factor extraction method based on the information criterion employed in our proposed approach. Comparison with PCA allows verification of the effectiveness of the newly introduced GGM-based factor.
\end{itemize}

\subsection{Evaluation with Sliding Testing Windows} \label{SubSec:Sliding-Testing-Windows}
This experiment evaluates performance of each method across multiple train–test splits spanning 2015 to 2024.
We defined a three-year training window (e.g., \texttt{2015/01}--\texttt{2017/12}) followed by a one-year test window (e.g., \texttt{2018/01}--\texttt{2018/12}).
The transformation matrix was estimated during the training period and was fixed during the test period.
We evaluated orthogonality of residual factors during the test period by the $\ell_1$ mean of cross-asset correlations, which measures absolute correlation magnitude, and the $\ell_2$ mean, which emphasizes the influence of larger correlations.
All experiments were repeated ten times, and mean results were reported for comparison.
\par
Table~\ref{Table:Results-Different-Period} summarizes results for the S\&P~500 and TOPIX~500, respectively. 
For both datasets, ICA and Whitening which do not address near-singular eigenstructures, consistently produced higher cross-asset correlations. 
In contrast, Shr.~Whitening with the OAS estimator yielded lower cross-asset correlations than standard whitening, confirming that handling near-singularity improves orthogonality for financial time series.
PCA residual factors further reduced correlations, indicating that removing high-variance components improves residual factor quality. 
The proposed method achieved the lowest correlations in every sample period, demonstrating effectiveness across diverse market regimes.

\begin{table*} 
    \centering
    \caption{Performance of the residual factors on the S\&P~500 and TOPIX~500. A higher Sharpe ratio (SR) is better, whereas lower Maximum Drawdown (MDD) and Conditional Value-at-Risk (CVaR) are better. The best values are shown in bold, and the second -best are underlined.}
    \begin{tabular}{l | l | ccc | ccc}
        \toprule
        & & & \textbf{S\&P} & & & \textbf{TPX} & \\
        data split & method & SR $\uparrow$ & MDD $\downarrow$ & CVaR $\downarrow$ & SR $\uparrow$ & MDD $\downarrow$ & CVaR $\downarrow$ \\
        \midrule
        Training & Baseline & 0.4015 & 0.0865 & 0.0066 & 0.0693 & 0.0934 & 0.0055 \\
        \texttt{2012/01-2016/12} & PCA & 0.7601 & 0.0531 & 0.0039 & 0.5242 & 0.0421 & 0.0032 \\
        Test & Shr.~Whitening & \underline{0.8164} & \underline{0.0481} & \underline{0.0039} & \underline{0.5687} & \underline{0.0399} & \underline{0.0030} \\
        \texttt{2017/01-2021/12} & \textbf{PCA+GGM} & \cellcolor{blue!15}\textbf{0.9022} & \cellcolor{blue!15}\textbf{0.0441} & \cellcolor{blue!15}\textbf{0.0034} & \cellcolor{blue!15}\textbf{0.5992} & \cellcolor{blue!15}\textbf{0.0378} & \cellcolor{blue!15}\textbf{0.0030} \\
        \midrule
        Training & Baseline & 0.4396 & 0.0895 & 0.0073 & 0.1184 & 0.0965 & 0.0059 \\
        \texttt{2013/01-2017/12} & PCA & \underline{0.6845} & 0.0619 & 0.0046 & \underline{0.6385} & 0.0422 & 0.0033 \\
        Test & Shr.~Whitening & 0.6455 & \underline{0.0549} & \underline{0.0045} & 0.5652 & \underline{0.0390} & \underline{0.0032} \\
        \texttt{2018/01-2022/12} & \textbf{PCA+GGM} & \cellcolor{blue!15}\textbf{0.7973} & \cellcolor{blue!15}\textbf{0.0522} & \cellcolor{blue!15}\textbf{0.0040} & \cellcolor{blue!15}\textbf{0.7041} & \cellcolor{blue!15}\textbf{0.0384} & \cellcolor{blue!15}\textbf{0.0030} \\
        \midrule
        Training & Baseline & 0.4470 & 0.0758 & 0.0074  & 0.1027 & 0.0819 & 0.0056 \\
        \texttt{2014/01-2018/12} & PCA  & \underline{0.7303} & 0.0582 & 0.0044 & \cellcolor{blue!15}\textbf{0.6549} & 0.0421 & 0.0030 \\
        Test & Shr.~Whitening & 0.6341 & \underline{0.0547} & \underline{0.0042}  & 0.6455 & \underline{0.0405} & \underline{0.0029}\\
        \texttt{2019/01-2023/12} & \textbf{PCA+GGM} & \cellcolor{blue!15}\textbf{0.7943} & \cellcolor{blue!15}\textbf{0.0495} & \cellcolor{blue!15}\textbf{0.0038}  & \underline{0.6489} & \cellcolor{blue!15}\textbf{0.0386} & \cellcolor{blue!15}\textbf{0.0026} \\
        \midrule
        Training & Baseline & 0.5780 & 0.0851 & 0.0075 & 0.2294 & 0.0748 & 0.0057 \\
        \texttt{2015/01-2019/12} & PCA & 0.6712 & 0.0562 & 0.0044 & \cellcolor{blue!15}\textbf{0.6608} & 0.0435 & 0.0031 \\
        Test & Shr.~Whitening & \underline{0.6895} & \underline{0.0503} & \underline{0.0042} & 0.5705 & \underline{0.0398} & \underline{0.0029} \\
        \texttt{2020/01-2024/12} & \textbf{PCA+GGM}  & \cellcolor{blue!15}\textbf{0.6924} & \cellcolor{blue!15}\textbf{0.0474} & \cellcolor{blue!15}\textbf{0.0039} & \underline{0.6021} & \cellcolor{blue!15}\textbf{0.0396} & \cellcolor{blue!15}\textbf{0.0028} \\
        \bottomrule
    \end{tabular}
    \label{Exp:Reversal-Result}
\end{table*} 
\subsection{Evaluation with Different Training Data Horizon} \label{SubSec:Different-Data-Size}
We next examine how predictive accuracy varies with the training window horizon.
We trained on data from \texttt{2015/01}--\texttt{2023/12} and tested from \texttt{2024/01}--\texttt{2024/12}, then progressively reduced the training window by one year at a time.
All experiments were repeated ten times. 
Motivated by performance observed in the previous experiment and to ensure figure readability, we used PCA and Shr.~Whitening as baselines for comparison in the present experiments.
Given that each year contains approximately 250 time series, when using only \texttt{2022/01}--\texttt{2023/12} or \texttt{2023/01}--\texttt{2023/12} as the training period, the sample covariance matrix becomes singular.
\begin{figure}[t]
  \centering
  \begin{subfigure}[t]{0.21\textwidth}
    \centering
    \includegraphics[width=\linewidth]{Figures/Cor-Matrix/PCA-Valid.pdf}
    \caption{PCA}
  \end{subfigure}
  \begin{subfigure}[t]{0.21\textwidth}
    \centering
    \includegraphics[width=\linewidth]{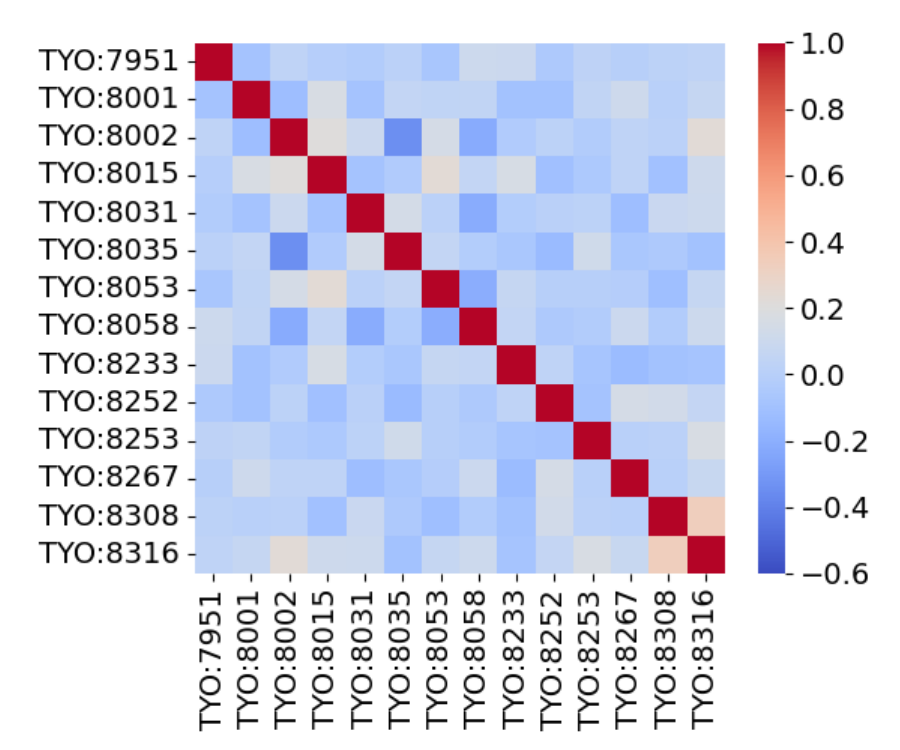}
    \caption{PCA+GGM}
  \end{subfigure}
  \caption{Ablation study of correlation matrices (Training period: \texttt{2020/01}--\texttt{2023/12}, Test period: \texttt{2024/01}--\texttt{2024/12}).}
  \label{Fig:Cor-Matrix}
\end{figure}
\par
Figure \ref{Fig:Changing-N-Results} shows the results.
In all figures, the horizontal axis marks the start of training. 
Positions further to the right indicate a shorter training window and fewer training samples.
Across all validation periods, the proposed method achieved lower values than PCA and Shr.~Whitening on both the $\ell_1$ and $\ell_2$ mean metrics of the correlation matrix.
This indicates that the residual factors maintain superior orthogonality regardless of the length of the training window.
Notably, the performance of Shr.~Whitening, which relies on shrinkage estimation, declined sharply as the training window shortened. 
This pattern implies that shrinkage-induced error increases as the covariance estimate becomes ill-conditioned, that is, near-singular. 
Although Shr.~Whitening outperformed PCA with large training windows, the proposed method consistently surpassed both across all settings, underscoring the effectiveness of the hierarchical GGM-based approach.
\par
Figure~\ref{Fig:Cor-Matrix} reports an ablation study comparing correlation matrices derived from models trained over the \texttt{2020/01}--\texttt{2023/12} window.
For clarity, only a subset of each correlation matrix is shown. 
In the PCA results shown in Figure~\ref{Fig:Cor-Matrix}~(a), a block-structured correlation pattern is evident.
This block corresponds to trading companies with tickers beginning with \texttt{TYO:80**}, indicating that residual factors from PCA did not fully eliminate within-sector correlations.
In contrast, Figure~\ref{Fig:Cor-Matrix}~(b) shows results of the proposed method, where the sector correlations were effectively removed, confirming effectiveness of the approach.
\par
Additional experiments with randomly selected stocks, as well as qualitative results for other methods, are provided in Appendix \ref{Apx:Additional-Results}.

\subsection{Performance Evaluation Using Trading Backtests}
Prior studies suggest that momentum strategies based on residual factors are effective \cite{blitz2011residual} because they remove time-varying common market factors. 
Guided by this premise, we compared methods under the assumption that orthogonal residual factors enhance momentum trading performance. 
In this setting, we implemented a short-term \emph{contrarian} strategy. 
Residual factors at time $t$ served as the signal. 
We then took the opposite position $d$ trading days later and closed it on day $t+d+1$.
Returns were computed from closing prices and rebound speeds likely vary across stocks.
Therefore, to construct a setup that is tradable in real-world transactions, we varied the execution lag $d$ from two to five days and compared average performance across these lags.
At each time step, we enforced a \emph{zero-investment portfolio} by taking equal-weighted long and short positions with matched dollar exposure.
\par
Consistent with the previous setup, we used historical price data for constituents of the S\&P~500 and TOPIX~500 from 2012 to 2024. 
We conducted a rolling-window evaluation with a 5-year training period followed by a 5-year test period, performing four -fold evaluation by shifting the window forward by one year at each step.
We considered PCA and Shr.~Whitening as comparative methods, and used the excess return relative to each index (i.e., S\&P~500 and TOPIX~500) as the \textbf{baseline}. 
We evaluated performance using the Sharpe ratio (\textbf{SR}), Maximum Drawdown (\textbf{MDD}), and Conditional Value at Risk (\textbf{CVaR}).
SR measures return per unit of risk and generally serves as the primary indicator of strategy efficiency. 
MDD is the largest peak-to-trough decline in cumulative returns.
CVaR captures tail risk, and we report the average of the worst 5\% of losses over the trading period.
Details of these metrics are provided in Appendix~\ref{Apx:Evaluation-Metrics}.
\par
Table \ref{Exp:Reversal-Result} reports trading results. 
Relative to the baseline using excess returns over each index, all methods based on residual factors achieved higher SR and lower MDD and CVaR in every sample period, confirming the advantage of residual factor approaches. 
The proposed method attained the highest SR among residual methods in most periods, indicating that its residual factors were particularly effective for contrarian trading. 
For MDD and CVaR, the proposed method consistently outperformed competing approaches across all cases, suggesting that removing common factor exposure yields better risk control. 
These findings align with prior evidence on the efficacy of strategies built on residual-return signals and show that our construction performs robustly across multiple windows.
Overall, these backtests suggest that the proposed factor decomposition framework is viable for practical trading.

\section{Conclusion}
In this paper, we present a method for estimating asset-specific residual factors from financial time series. 
The approach applies PCA followed by GGM that enforces the MTP$_2$ constraint. 
We show theoretically that incorporating MTP$_2$ GGM yields residual factors with more orthogonal cross-asset correlations across assets than those produced by PCA alone. 
Empirical results corroborate this analysis and demonstrate effectiveness of the method.
Backtest results indicate that using residual factors can substantially reduce risk by removing market-wide common factors. 
Future work includes detailed analyses of residual factors, methods that account for temporal dependence in financial time series, and improved scalability.

\clearpage



\bibliographystyle{named}
\bibliography{ijcai26}

\clearpage

\appendix
\section{Proof of Theorem \ref{Thrm:S-Z-Projection}} \label{Apx:Proof-S-Z-Projection}
First, we present the result for $\Pi_{\mathcal{S}_{++}^{N}}(\cdot)$.  
By exploiting invariance of the Frobenius norm under orthogonal transformations, we obtain
\begin{align*}
    \|\bm{\Lambda} - \mathbf{Y}\|_{\mathrm{F}}^{2}
        &= \|\mathbf{\Omega} - \mathbf{Q}^{\top} \mathbf{Y} \mathbf{Q}\|_{\mathrm{F}}^{2}.
\end{align*}
Since $\mathbf{\Omega}$ is a diagonal matrix, the minimizer $\mathbf{Y}$ must satisfy that the off-diagonal entries of $\mathbf{Q}^{\top} \mathbf{Y} \mathbf{Q}$ are zero.  
Furthermore, imposing that $\mathbf{Y}$ is positive definite implies that the minimizer satisfies $\mathbf{Y}$ satisfies $\mathbf{Q}^{\top} \mathbf{Y} \mathbf{Q} = \mathbf{\Omega}_{+}$.  
It follows that $\mathbf{Y} = \mathbf{Q}\mathbf{\Omega}_{+}\mathbf{Q}^{\top}$.

Next, we present the result for $\Pi_{\mathcal{Z}^{N}}(\cdot)$.  
From $\|\bm{\Lambda} - \mathbf{Y}\|_{\mathrm{F}}^{2} = \sum_{i,j} (\Lambda_{i,j} - y_{i,j})^{2}$, it is clear that at the point satisfying Eq.~\eqref{Eq:Prox-Z-Matrix},  
diagonal elements satisfy $y_{i,j} = \Lambda_{i,j}$, while off-diagonal elements satisfy $y_{i,j} = \min(\Lambda_{i,j}, 0)$.  
Therefore, Eq.~\eqref{Eq:Closed-Form-Prox-Z-Matrix} holds.

\section{Proof of Proposition \ref{Prop:Cov-R}}
\label{Apx:Proof-Prop-Cov-R}
By substituting $\mathbf{R}$, we obtain
\begin{align*}
    &\mathrm{cov}(\mathbf{R}, \mathbf{Z}) \\
    &= \mathbf{D}^{-1} \bm{\Lambda} \; \mathrm{cov}(\mathbf{Z}) \\
    &= \mathbf{D}^{-1}.
\end{align*}

\section{Proof of Proposition \ref{Prop:Cor-R-is-smaller-than-Cor-Z}}
\label{Apx:Proof-Cor-R-is-smaller-than-Cor-Z}
We first define terminology used in the proof. 
Let $\mathbf{Z}$ denote the return matrix obtained via PCA transformation $\mathbf{W}_{\mathrm{PCA}}$ as $\mathbf{Z}=\mathbf{W}_{\mathrm{PCA}}\mathbf{X}$, and let $\mathbf{R}=\mathbf{W}\mathbf{X}$ denote the residual factor matrix obtained via the final transformation 
\[
\mathbf{W} = \mathbf{W}_{\mathrm{GGM}}\mathbf{W}_{\mathrm{PCA}} = \mathbf{D}^{-1}\bm{\Lambda}\mathbf{W}_{\mathrm{PCA}},
\]  
where $\bm{\Lambda}$ is the precision matrix estimated by the MTP$_2$ GGM model in Eq. \eqref{Eq:MTP2-Problem}.
We define the inverse of the estimated precision matrix as $\bm{\Sigma} = \bm{\Lambda}^{-1}$, and let $C(\cdot)$ denote the normalization function for a covariance matrix.
Accordingly, we denote the correlation matrix of PCA residual factors as $C(\mathbf{Z}\mathbf{H}\mathbf{Z}^{\top})$ and that of residual factors as $C(\mathbf{R}\mathbf{H}\mathbf{R}^{\top})$ where $\mathbf{H}=\mathbf{I}-\frac{1}{T}\mathbf{1}\mathbf{1}^{\top}$.
We denote by $\mathrm{cov}(\cdot)$ the covariance matrix of a given variable, and $\mathrm{cor}(\cdot)$ its correlation matrix.  
For matrices $\mathbf{A}$ and $\mathbf{B}$, notation $\mathbf{A}\geq 0$ indicates that all elements of $\mathbf{A}$ are nonnegative, while $\mathbf{A}\geq\mathbf{B}$ indicates that $a_{i,j} \geq b_{i,j}$ holds  for all $i,j$.
Finally, let $\mathcal{IM}$ denote the set of inverse \textit{M}-matrices, i.e., matrices whose inverse is an \textit{M}-matrix $\mathcal{M}$.
\par
First, we present the following lemmas holding for \textit{M}-matrices and inverse \textit{M}-matrices:
\begin{lemma}\label{Lemma:IM-Positivity}
    All elements of an inverse \textit{M}-matrix are nonnegative~\cite{markham1972nonnegative}. 
    In other words,
    \begin{align*}
        \bm{\Sigma} \in \mathcal{IM} \ \rightarrow \ \bm{\Sigma} \geq 0.
    \end{align*}
\end{lemma}
\begin{lemma}\label{Lemma:M-Causal-Closure}
    Every principal submatrix of an \textit{M}-matrix is itself an \textit{M}-matrix~\cite{plemmons1977m}.
\end{lemma}
\begin{lemma}\label{Lemma:IM-Causal-Closure}
    Every principal submatrix of an inverse \textit{M}-matrix is itself an inverse \textit{M}-matrix~\cite{markham1972nonnegative}.
\end{lemma}
\begin{lemma}\label{Lemma:IM-Corr-InEq}
    For $\bm{\Sigma}, \bm{\Gamma} \in \mathcal{IM}$, the following holds~\cite{karlin1983m}:
    \begin{align*}
        \bm{\Gamma}^{-1} \leq \bm{\Sigma}^{-1} \ \rightarrow \ C(\bm{\Sigma}) \leq C(\bm{\Gamma}).
    \end{align*}
\end{lemma}
\par
Using the above lemmas, we prove Proposition~\ref{Prop:Cor-R-is-smaller-than-Cor-Z} in the main text according to the following procedure.
\begin{mdframed}[linewidth=1pt, backgroundcolor=white, roundcorner=20pt]
    \textbf{Proof sketch of Proposition~\ref{Prop:Cor-R-is-smaller-than-Cor-Z}}
    \begin{enumerate}
        \item From Lemma~\ref{Lemma:Cor-R-Is-Partial-Cor}, absolute correlations in $\mathbf{R}$ equal partial correlations in $\mathbf{Z}$:
        \[
            |\mathrm{cor}(\mathbf{R})| = \mathrm{pcor}(\mathbf{Z}).
        \]
        \item From Lemma~\ref{Lemma:Cond-Var}, for $\mathbf{Z}_{1} = [\mathbf{z}_{i}, \mathbf{z}_{j}]^{\top}$ after reordering variables,
        \begin{align*}
            \mathrm{cor}(\mathbf{Z}_{1}) &= C(\bm{\Sigma}_{11}), \\
            \mathrm{pcor}(\mathbf{Z}_{1}) &= C\!\left(\bm{\Sigma}_{11} - \bm{\Sigma}_{12} \bm{\Sigma}_{22}^{-1} \bm{\Sigma}_{21}\right).
        \end{align*}
        \item Lemma~\ref{Lemma:Pcor_inv-vs-Cor_inv} yields
        \[
            \left(\bm{\Sigma}_{11} - \bm{\Sigma}_{12} \bm{\Sigma}_{22}^{-1} \bm{\Sigma}_{21}\right)^{-1} \geq \bm{\Sigma}_{11}^{-1}.
        \]
        \item By Lemma~\ref{Lemma:IM-Corr-InEq}, this implies
        \[
            C\!\left(\bm{\Sigma}_{11} - \bm{\Sigma}_{12} \bm{\Sigma}_{22}^{-1} \bm{\Sigma}_{21}\right) \leq C(\bm{\Sigma}_{11}).
        \]
        \item Applying the above to all partitions $\mathbf{Z}_{1}=[\mathbf{z}_{i},\mathbf{z}_{j}]$ proves Proposition~\ref{Prop:Cor-R-is-smaller-than-Cor-Z}.
    \end{enumerate}
\end{mdframed}

\par
Hereafter, we prove each lemma required above.
\begin{lemma}\label{Lemma:Cor-R-Is-Partial-Cor}
The magnitude of correlation coefficients of $\mathbf{R}$ coincides with the partial correlations of $\mathbf{Z}$:
\begin{equation}
    \begin{aligned}
         |\mathrm{cor}(\mathbf{R})| &= \mathrm{pcor}(\mathbf{Z}).
    \end{aligned}
    \label{Eq:Cor-Rij}
\end{equation}
\end{lemma}

\begin{proof}
    First, from $\mathrm{cov}(\mathbf{R}) = \mathrm{cov}(\mathbf{D}^{-1} \bm{\Lambda} \mathbf{Z}) 
    = \mathbf{D}^{-1} \bm{\Lambda} \mathbf{D}^{-1}$, we have 
    $[\mathrm{cov}(\mathbf{R})]_{i,j} = \frac{\Lambda_{i,j}}{\Lambda_{i,i} \Lambda_{j,j}}$.  
    Therefore, the $(i,j)$-th entry of the correlation matrix $[\mathrm{cor}(\mathbf{R})]_{i,j}$ can be written as
    \begin{equation}
        \begin{aligned}
             [\mathrm{cor}(\mathbf{R})]_{i,j} &= \frac{\Lambda_{i,j}}{\sqrt{\Lambda_{i,i}} \sqrt{\Lambda_{j,j}}}.
        \end{aligned}
    \end{equation}
    Moreover, the partial correlation $\rho_{p}(\mathbf{z}_{i}, \mathbf{z}_{j})$ of $\mathbf{Z}$ is given by 
    $-\frac{\Lambda_{i,j}}{\sqrt{\Lambda_{i,i}} \sqrt{\Lambda_{j,j}}}$.  
    Since $\bm{\Lambda} \in \mathcal{M}$, we have $\Lambda_{i,j} \leq 0$.  
    Thus,
    \begin{align*}
        \big| [\mathrm{cor}(\mathbf{R})]_{i,j} \big| = \rho_{p}(\mathbf{z}_{i}, \mathbf{z}_{j}).
    \end{align*}
\par
From the above, letting $\mathrm{pcor}(\bm{\Sigma})$ denote the partial correlation matrix of the covariance matrix $\bm{\Sigma}$, we can write $|\mathrm{cor}(\mathbf{R})| = \mathrm{pcor}(\mathbf{Z})$. 
Therefore, proving Proposition~\ref{Prop:Cor-R-is-smaller-than-Cor-Z} is equivalent to showing that $\mathrm{pcor}(\mathbf{Z}) \leq \mathrm{cor}(\mathbf{Z})$.
\end{proof}
\begin{lemma}\label{Lemma:Cond-Var}
    First, partition $\mathbf{Z}$ as
    \[
    \mathbf{Z}=
    \begin{bmatrix}
        \mathbf{Z}_{1} \\
        \mathbf{Z}_{2}
    \end{bmatrix},
    \]
    and let $\bm{\Sigma} = \mathrm{cov}(\mathbf{Z})$ be given by
    \begin{align*}
    \bm{\Sigma} &= \begin{bmatrix}
        \mathrm{cov}(\mathbf{Z}_{1}) & \mathrm{cov}(\mathbf{Z}_{1},\mathbf{Z}_{2}) \\
        \mathrm{cov}(\mathbf{Z}_{2}, \mathbf{Z}_{1}) & \mathrm{cov}(\mathbf{Z}_{2})
    \end{bmatrix} \\
    &= \begin{bmatrix}
        \bm{\Sigma}_{11} & \bm{\Sigma}_{12} \\
        \bm{\Sigma}_{21} & \bm{\Sigma}_{22}
    \end{bmatrix},
    \end{align*}
    where $\bm{\Sigma}_{12} = \bm{\Sigma}_{21}^\top$. In this case, the following holds:
    \begin{equation}
        \begin{aligned}
            \mathrm{cor}(\mathbf{Z}_{1}) &= C(\bm{\Sigma}_{11}), \\
            \mathrm{pcor}(\mathbf{Z}_{1}) &= C\left(\bm{\Sigma}_{11} - \bm{\Sigma}_{12} \bm{\Sigma}_{22}^{-1} \bm{\Sigma}_{21}\right).
        \end{aligned}
        \label{Eq:Lemma-Pcor}
    \end{equation}
\end{lemma}
\begin{proof}
    By definition of partial correlation, the partial correlation is the correlation coefficient conditioned on all variables except $\mathbf{Z}_{1}$. 
    From properties of the conditional distribution of a multivariate normal vector, Lemma~\ref{Lemma:Cond-Var} follows immediately.
\end{proof}
Here, by taking $\mathbf{Z}_{1} = [\mathbf{z}_{i}, \mathbf{z}_{j}]^{\top}$ and reordering of variables, the result for the block-matrix formulation applies directly to correlation between $\mathbf{z}_{i}$ and $\mathbf{z}_{j}$.
\par
From Lemma~\ref{Lemma:IM-Causal-Closure}, it follows that $\bm{\Sigma}_{11} \in \mathcal{IM}$.  
Finally, we prove the following lemma for inverse \textit{M}-matrices.
\begin{lemma}\label{Lemma:Pcor_inv-vs-Cor_inv}
    If $\bm{\Sigma} \in \mathcal{IM}$, then  
    $\bm{\Sigma}_{11} - \bm{\Sigma}_{12} \bm{\Sigma}_{22}^{-1} \bm{\Sigma}_{21} \in \mathcal{IM}$ 
    and
    \begin{align*}
        \left(\bm{\Sigma}_{11} - \bm{\Sigma}_{12} \bm{\Sigma}_{22}^{-1} \bm{\Sigma}_{21}\right)^{-1} \geq \bm{\Sigma}_{11}^{-1}
    \end{align*}
    hold.
\end{lemma}
Here, we present the following lemma, which is useful for proving Lemma~\ref{Lemma:Pcor_inv-vs-Cor_inv}.
\begin{lemma}\label{Lemma:Positivity-of-Pcor}
    Let $\bm{\Gamma} \in \mathcal{IM}^{N+1}$ be given by
    \begin{align*}
        \bm{\Gamma} =
        \begin{bmatrix}
            \gamma & \bm{\gamma}^{\top} \\
            \bm{\gamma} & \bm{\Gamma}'
        \end{bmatrix}.
    \end{align*}
    Then, the following holds:
    \begin{align*}
        (\bm{\Gamma}')^{-1} \bm{\gamma} \geq 0 .
    \end{align*}
\end{lemma}
\begin{proof}
    Let $\bm{\Xi} = \bm{\Gamma}^{-1}$, $\bm{\Xi}' = (\bm{\Gamma}')^{-1}$, and $\bm{\gamma} = [\gamma_{1}, \gamma_{2}, \ldots, \gamma_{N}]^{\top}$.  
    From the cofactor expansion theorem, letting $\tilde{\bm{\Gamma}}_{i,j}'$ denote the $(i,j)$-th adjugate matrix of $\bm{\Gamma}'$, we have
    \[
    [\bm{\Xi}']_{i,j} = \frac{1}{\det(\bm{\Gamma}')} (-1)^{i+j} \det(\tilde{\bm{\Gamma}}_{i,j}'),
    \]
    and therefore
    \begin{align*}
        [\bm{\Xi}' \bm{\gamma}]_{i}
        &= \sum_{j=1}^{N} \frac{1}{\det(\bm{\Gamma}')} (-1)^{i+j} \gamma_{j} \det(\tilde{\bm{\Gamma}}_{i,j}') \\
        &= \frac{(-1)^{i}}{\det(\bm{\Gamma}')} \sum_{j=1}^{N} (-1)^{j} \gamma_{j} \det(\tilde{\bm{\Gamma}}_{i,j}').
    \end{align*}
    Here, let $\hat{\bm{\Gamma}}'_{i,:}$ denote the matrix obtained by removing the $i$-th row of $\bm{\Gamma}'$.  
    Then, by the cofactor expansion theorem, the following equation holds:
    \[
    \sum_{j=1}^{N} (-1)^{j} \gamma_{j} \det(\tilde{\bm{\Gamma}}_{i,j}') 
    = -\det\!\left(
    \begin{bmatrix}
        \bm{\gamma}^{\top} \\
        \hat{\bm{\Gamma}}'_{i,:}
    \end{bmatrix}
    \right).
    \]
    Furthermore, from the expression of an inverse matrix using cofactors, the $(i+1,1)$-th element of $\bm{\Xi}$ can be written as
    \[
    [\bm{\Xi}]_{i+1,1} = \frac{1}{\det(\bm{\Gamma})} (-1)^{i+2} \gamma_{i} 
    \det\!\left(
    \begin{bmatrix}
        \bm{\gamma}^{\top} \\
        \hat{\bm{\Gamma}}'_{i,:}
    \end{bmatrix}
    \right).
    \]
    Substituting this into the previous expression, we obtain
    \begin{align*}
        [\bm{\Xi}' \bm{\gamma}]_{i} 
        &= -\frac{(-1)^{i}}{\det(\bm{\Gamma}')} 
           \det\!\left(
           \begin{bmatrix}
               \bm{\gamma}^{\top} \\
               \hat{\bm{\Gamma}}'_{i,:}
           \end{bmatrix}
           \right) \\
        &= -\frac{\det(\bm{\Gamma})}{\det(\bm{\Gamma}')} \cdot \frac{[\bm{\Xi}]_{i+1,1}}{\gamma_{i}}.
    \end{align*}
    Since $\bm{\Gamma} \in \mathcal{IM}^{N+1}$, Lemma~\ref{Lemma:IM-Positivity} implies that $\gamma_{i} \geq 0$.  
    Moreover, because $\bm{\Xi} \in \mathcal{M}^{N+1}$, we have $[\bm{\Xi}]_{i+1,1} \leq 0$.  
    As both $\bm{\Gamma}$ and $\bm{\Gamma}'$ are positive definite, it follows that $\det(\bm{\Gamma}) > 0$ and $\det(\bm{\Gamma}') > 0$.  
    Therefore, $[\bm{\Xi}' \bm{\Gamma}]_{i} \geq 0$, and hence $\bm{\Xi}' \bm{\Gamma} \geq 0$.
\end{proof}
By using Lemma~\ref{Lemma:Positivity-of-Pcor}, we now prove Lemma~\ref{Lemma:Pcor_inv-vs-Cor_inv}.
\begin{proof}
    First, by Woodbury’s identity, we can expand as
    \begin{align}
        &(\bm{\Sigma}_{11} - \bm{\Sigma}_{12} \bm{\Sigma}_{22}^{-1} \bm{\Sigma}_{21})^{-1} \notag \\
        &= \bm{\Sigma}_{11}^{-1} + \bm{\Sigma}_{11}^{-1} \bm{\Sigma}_{12} \bm{\mathrm{S}}_{22}^{-1} \bm{\Sigma}_{21} \bm{\Sigma}_{11}^{-1}, \label{Eq:Woodbury}
    \end{align}
    where we define 
    \[
    \mathbf{S}_{22} = \bm{\Sigma}_{22} - \bm{\Sigma}_{21} \bm{\Sigma}_{11}^{-1} \bm{\Sigma}_{12}.
    \]
    From the expression for the Schur complement, the inverse of a block matrix can be written as
    \begin{align*}
        \bm{\Sigma}^{-1} =
        \begin{bmatrix}
            (\bm{\Sigma}_{11} - \bm{\Sigma}_{12} \bm{\Sigma}_{22}^{-1} \bm{\Sigma}_{21})^{-1} & -\bm{\Sigma}_{11}^{-1} \bm{\Sigma}_{12} \mathbf{S}_{22}^{-1} \\
            -\mathbf{S}_{22}^{-1} \bm{\Sigma}_{21} \bm{\Sigma}_{11}^{-1} & \mathbf{S}_{22}^{-1}
        \end{bmatrix}.
    \end{align*}
    From this, we have $(\bm{\Sigma}_{11} - \bm{\Sigma}_{12} \bm{\Sigma}_{22}^{-1} \bm{\Sigma}_{21})^{-1} \in \mathcal{M}$ (Lemma~\ref{Lemma:M-Causal-Closure}), and thus $\bm{\Sigma}_{11} - \bm{\Sigma}_{12} \bm{\Sigma}_{22}^{-1} \bm{\Sigma}_{21} \in \mathcal{IM}$.  
    \par
    In the following, we evaluate the sign of the second term in Eq.~\eqref{Eq:Woodbury}.    
    First, since $\bm{\Sigma} \in \mathcal{IM}$, it follows that $\bm{\Sigma}^{-1} \in \mathcal{M}$.  
    From the properties of an \textit{M}-matrix, the off-diagonal elements are nonpositive, and hence 
    \[
    \bm{\Sigma}_{11}^{-1} \bm{\Sigma}_{12} \bm{\mathrm{S}}_{22}^{-1} \geq 0.
    \]
    Next, we evaluate the sign of $\bm{\Sigma}_{21} \bm{\Sigma}_{11}^{-1}$.  
    Since all principal submatrices are \textit{M}-matrices (Lemma~\ref{Lemma:IM-Causal-Closure}), Lemma~\ref{Lemma:Positivity-of-Pcor} applies to each row, which yields 
    \[
    \bm{\Sigma}_{21} \bm{\Sigma}_{11}^{-1} \geq 0.
    \]
    Therefore, from $\bm{\Sigma}_{11}^{-1} \bm{\Sigma}_{12} \bm{\mathrm{S}}_{22}^{-1} \bm{\Sigma}_{21} \bm{\Sigma}_{11}^{-1} \geq 0$, we obtain
    \[
    (\bm{\Sigma}_{11} - \bm{\Sigma}_{12} \bm{\Sigma}_{22}^{-1} \bm{\Sigma}_{21})^{-1} \geq \bm{\Sigma}_{11}^{-1}.
    \]
\end{proof}

\section{Additional Results}
\label{Apx:Additional-Results}

\begin{figure*}[t]
  \centering
  \begin{subfigure}[t]{0.49\textwidth}
    \centering
    \includegraphics[width=\linewidth]{Figures/All-Assets/SP500-Changing-Datasize.pdf}
    \caption{All Assets (S\&P500)}
  \end{subfigure}
  \hfill
  \begin{subfigure}[t]{0.49\textwidth}
    \centering
    \includegraphics[width=\linewidth]{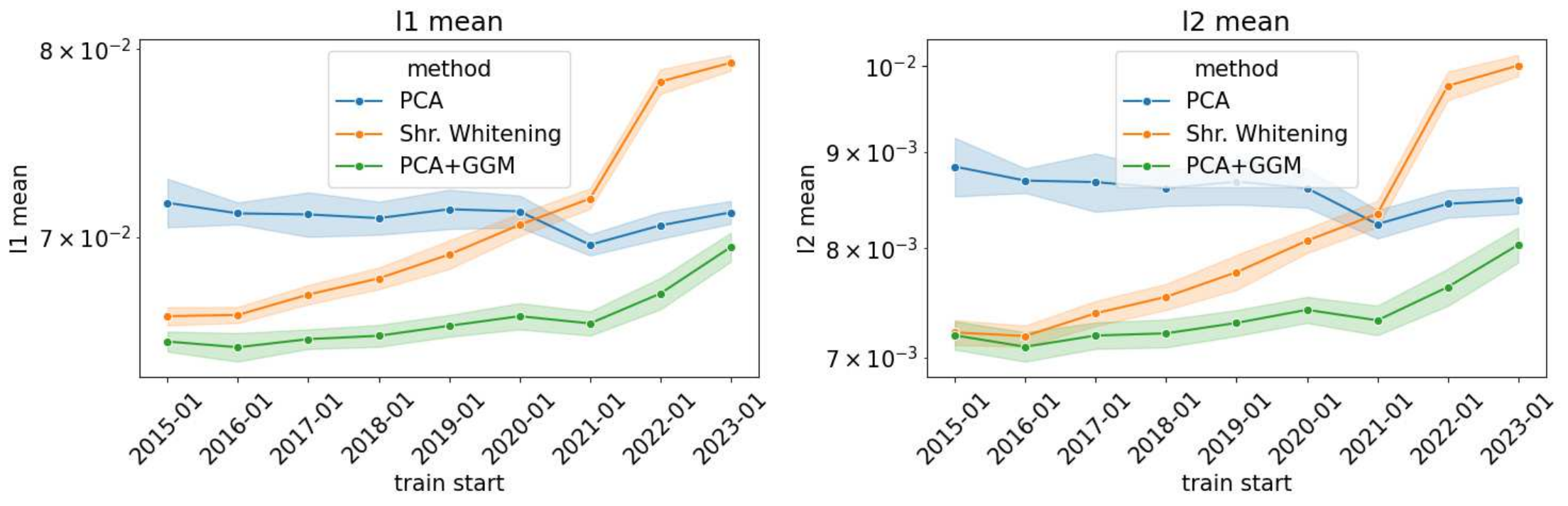}
    \caption{Random Assets (S\&P500)}
  \end{subfigure}
  \centering
  \begin{subfigure}[t]{0.49\textwidth}
    \centering
    \includegraphics[width=\linewidth]{Figures/All-Assets/TOPIX-Changing-Datasize.pdf}
    \caption{All Assets (TOPIX500)}
  \end{subfigure}
  \hfill
  \begin{subfigure}[t]{0.49\textwidth}
    \centering
    \includegraphics[width=\linewidth]{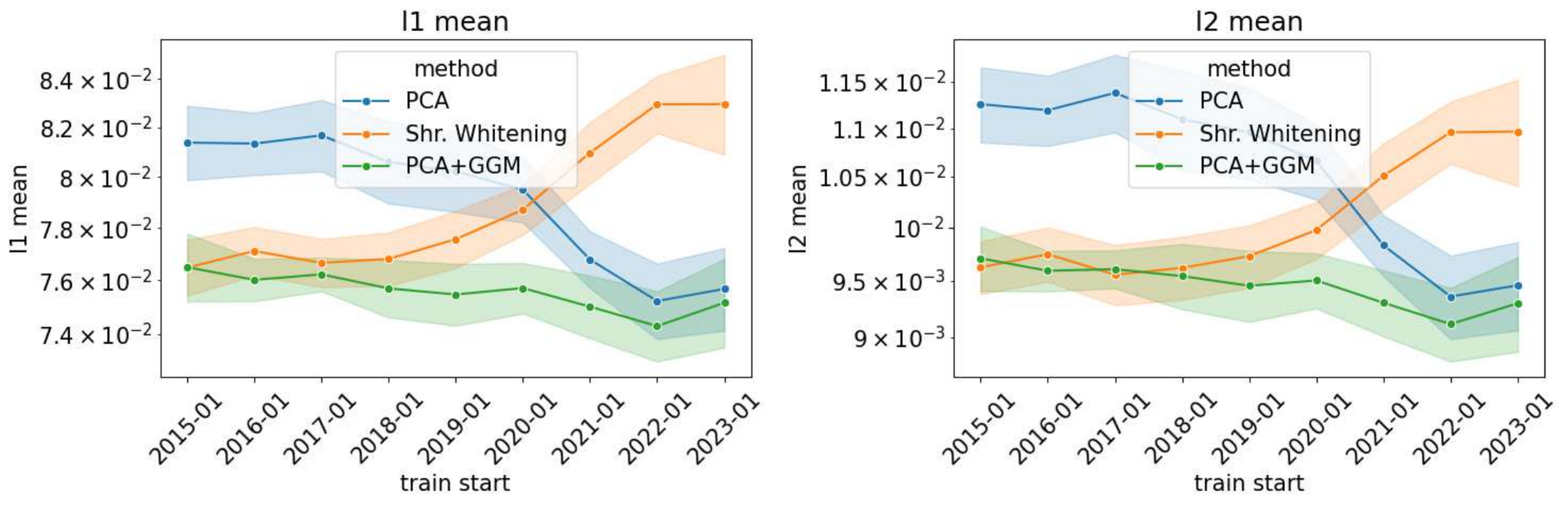}
    \caption{Random Assets (TOPIX500)}
  \end{subfigure}
  
  \caption{The $\ell_1$ and $\ell_2$ mean and standard deviation of the cross correlations as the length of training samples varies for the S\&P~500 and TOPIX~500. The horizontal axis indicates the timestamp at which the training period starts. The left column reports results using all constituents; the right column reports results when 300 constituents are randomly selected. All results are averaged over 10 trials.}
  \label{Fig:Apx-Changing-N-Results}
\end{figure*}

\begin{figure*}[t]
  \centering
  \begin{subfigure}[t]{0.19\textwidth}
    \centering
    \includegraphics[width=\linewidth]{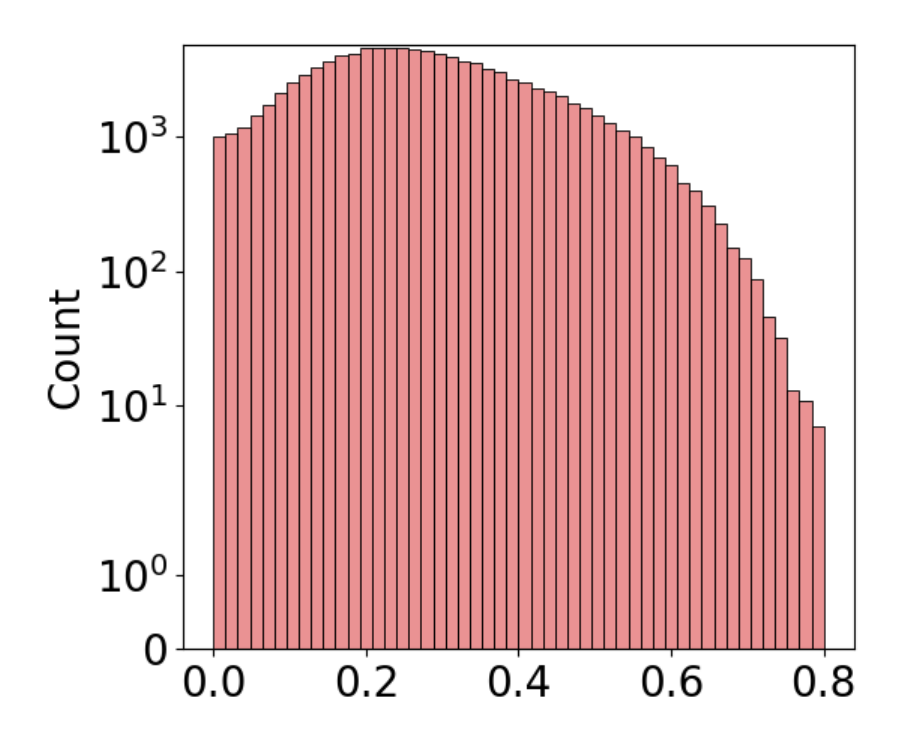}
    \caption{Raw Correlation}
  \end{subfigure}
  \begin{subfigure}[t]{0.19\textwidth}
    \centering
    \includegraphics[width=\linewidth]{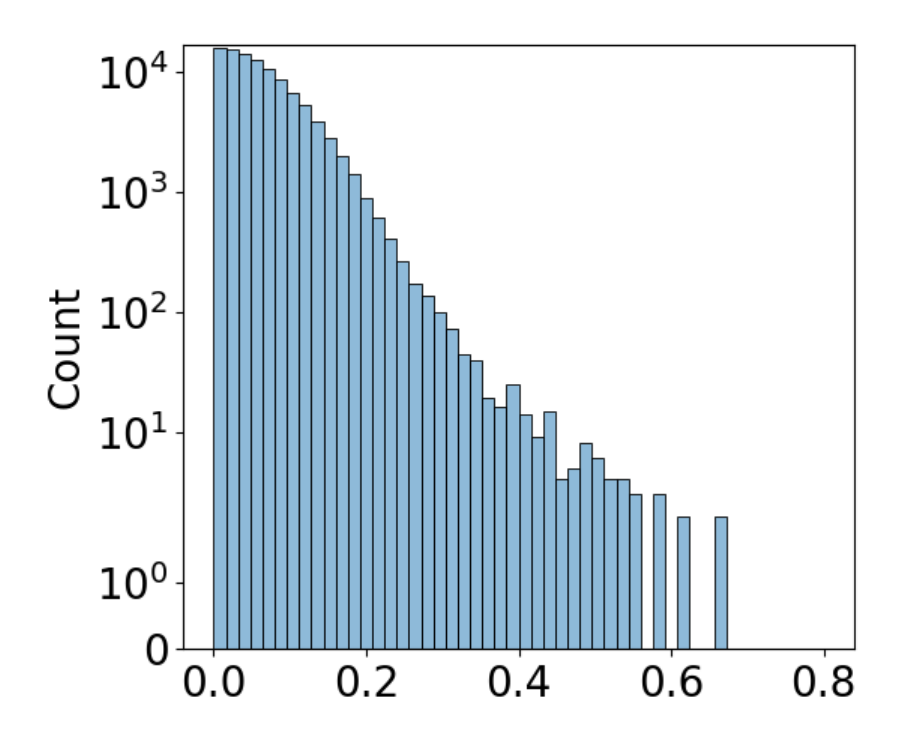}
    \caption{PCA}
  \end{subfigure}
  \begin{subfigure}[t]{0.19\textwidth}
    \centering
    \includegraphics[width=\linewidth]{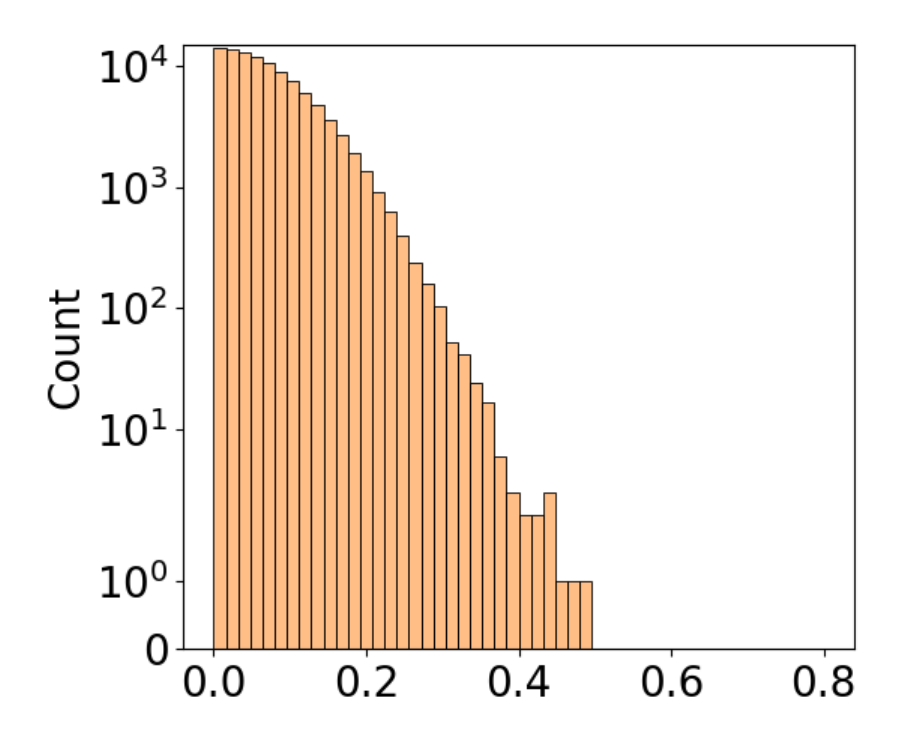}
    \caption{Shr. Whitening}
  \end{subfigure}
  \begin{subfigure}[t]{0.19\textwidth}
    \centering
    \includegraphics[width=\linewidth]{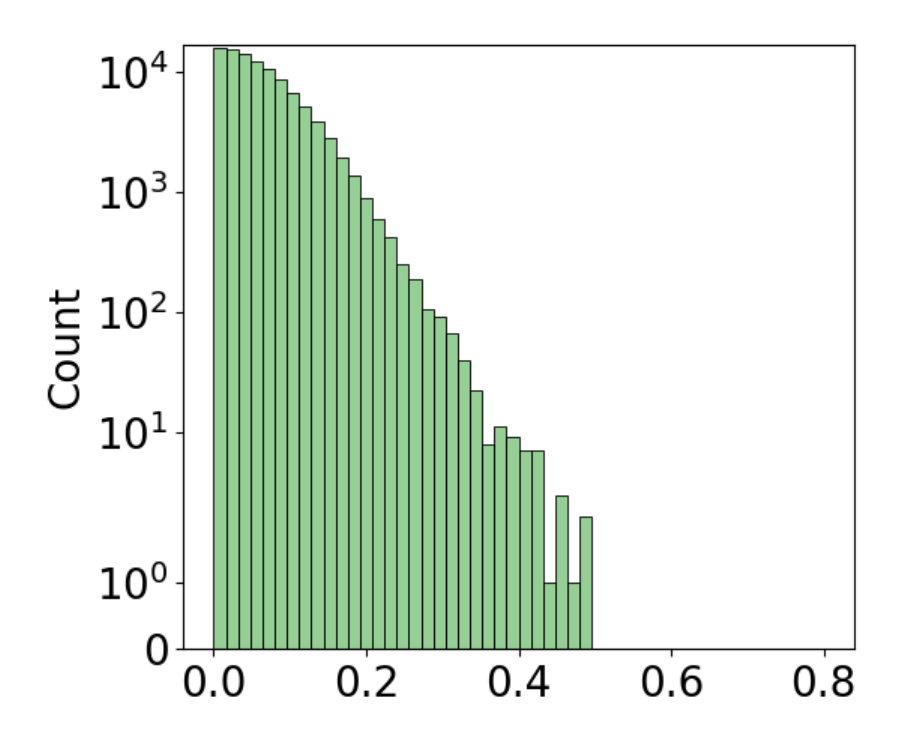}
    \caption{PCA+GGM}
  \end{subfigure}
  \begin{subfigure}[t]{0.19\textwidth}
    \centering
    \includegraphics[width=\linewidth]{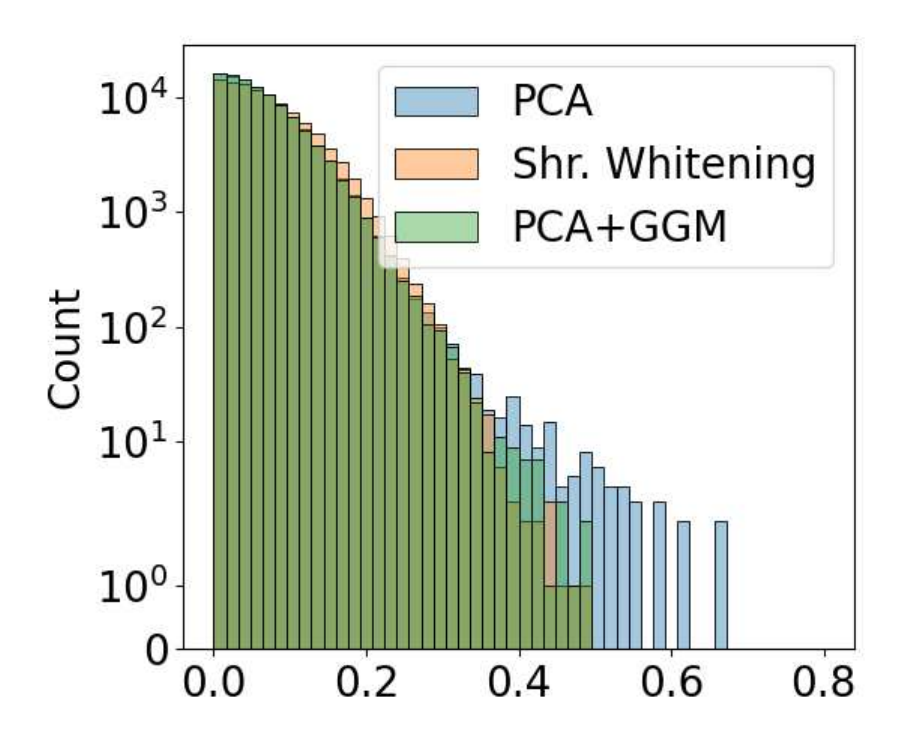}
    \caption{All methods}
  \end{subfigure}
  \caption{Histogram of the absolute values of pairwise correlations among TOPIX~500 constituents, with the training period set to \texttt{2023/01}--\texttt{2023/12} and the validation period to \texttt{2024/01}--\texttt{2024/12}. Note that the $y$-axis is on a logarithmic scale.}
  \label{Fig:Apx-Histgrams}
\end{figure*}

\subsection{Results with Randomly Selected Datasets}
This subsection presents additional results for the experiments in Section \ref{SubSec:Different-Data-Size}.
As in those experiments, we used \texttt{2015/01}--\texttt{2023/12} as the training period and \texttt{2024/01}--\texttt{2024/12} as the test period.
We progressively shortened the training period by one year.
To further assess the effect of the number of periods, we evaluated a setting in which orthogonalization is performed on $300$ randomly selected stocks, thereby increasing the ratio of time periods to the number of stocks.
\par
Figure~\ref{Fig:Apx-Changing-N-Results} presents the experimental results.
The figure shows that, across many evaluation periods, the proposed method achieved lower correlations than PCA and Shr.~Whitening on both metrics— namely the $\ell_1$- and $\ell_2$-means of the correlation matrix.
The residual factors produced by the proposed method exhibited strong orthogonality regardless of the length of the training window.
A similar pattern was observed under random selection as well, where the proposed method achieved superior orthogonality in most cases.
\par
To further examine the details, we present in Figure~\ref{Fig:Apx-Histgrams} the histogram of the absolute correlations of residual factors for the TOPIX~500 constituents. 
The training period is the one year from \texttt{2023/01}--\texttt{2023/12}, and note that the $y$-axis is on a logarithmic scale. 
Figure~\ref{Fig:Apx-Histgrams}~(a) shows that correlations among raw return series are predominantly positive, peaking around $0.2$, whereas for residual-return series the peak shifts toward $0$.
Comparing PCA in Figure~\ref{Fig:Apx-Histgrams}~(b) with the proposed method in Figure~\ref{Fig:Apx-Histgrams}~(d), many absolute correlations above $0.4$ remained under PCA, suggesting the presence of correlations that the information criterion failed to remove. 
Shr.~Whitening in Figure~\ref{Fig:Apx-Histgrams}~(c) does not leave very large absolute correlations.
However, Figure~\ref{Fig:Apx-Histgrams}~(e) indicates that, relative to the proposed method, a larger mass of correlations in the $0.2$--$0.4$ range remains. 
The proposed method in Figure~\ref{Fig:Apx-Histgrams}~(d) succeeds in removing the large absolute correlations observed under PCA, indicating that the GGM-based procedure achieved stronger orthogonalization.

\begin{figure*}[t]
  \centering
  \begin{subfigure}[t]{0.24\textwidth}
    \centering
    \includegraphics[width=\linewidth]{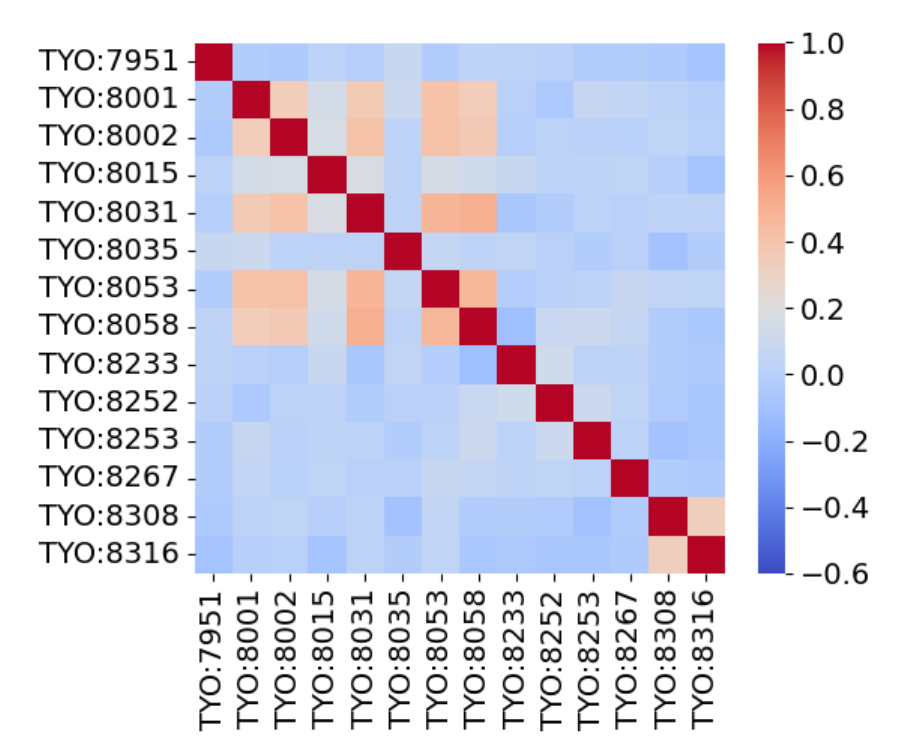}
    \caption{PCA (Train)}
  \end{subfigure}
  \begin{subfigure}[t]{0.24\textwidth}
    \centering
    \includegraphics[width=\linewidth]{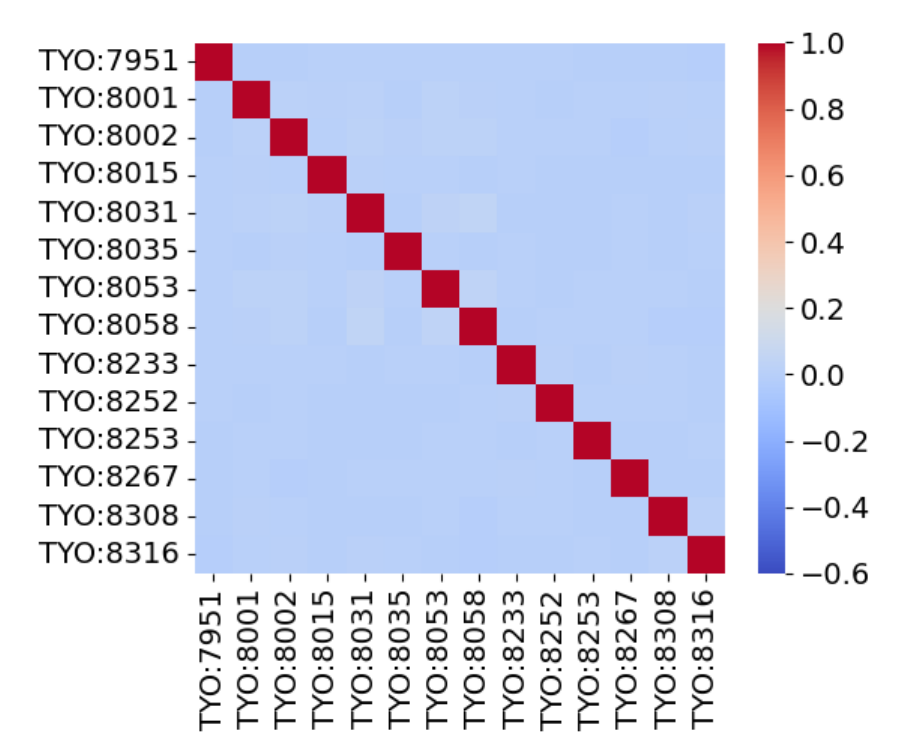}
    \caption{Shr. Whitening (Train)}
  \end{subfigure}
  \begin{subfigure}[t]{0.24\textwidth}
    \centering
    \includegraphics[width=\linewidth]{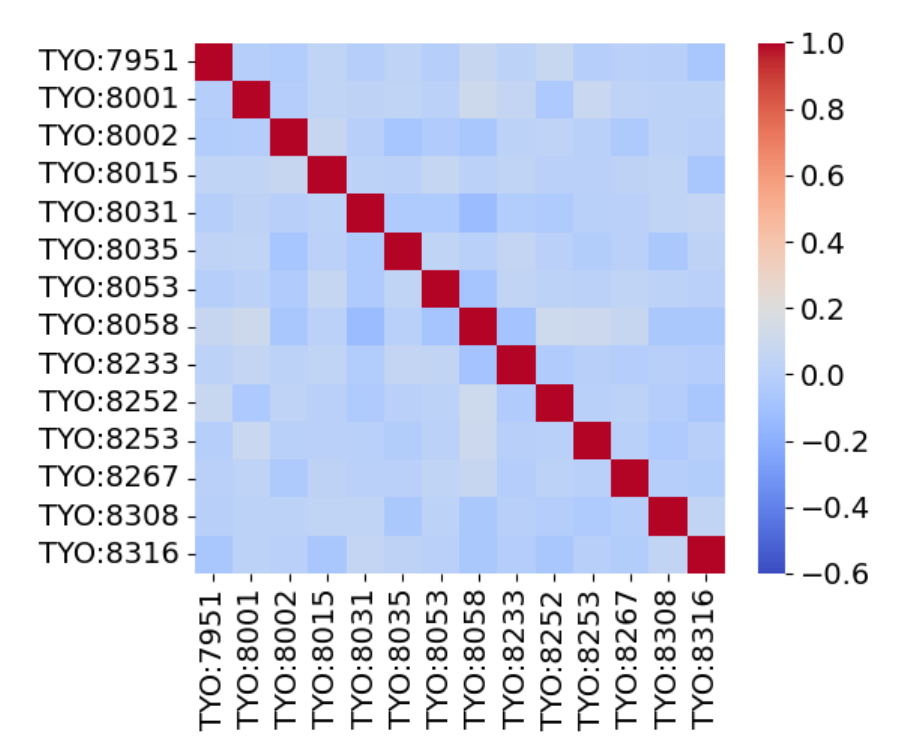}
    \caption{PCA+GGM (Train)}
  \end{subfigure}
  \begin{subfigure}[t]{0.24\textwidth}
    \centering
    \includegraphics[width=\linewidth]{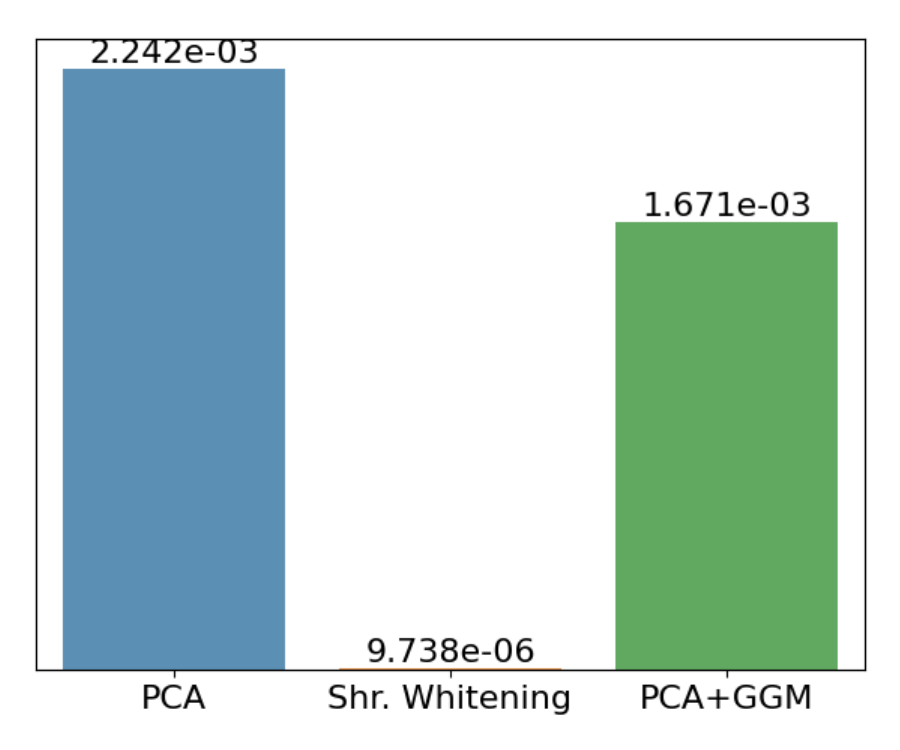}
    \caption{$\ell_2$ mean (Train)}
  \end{subfigure}
  \begin{subfigure}[t]{0.24\textwidth}
    \centering
    \includegraphics[width=\linewidth]{Figures/Cor-Matrix/PCA-Valid.pdf}
    \caption{PCA (Valid)}
  \end{subfigure}
  \begin{subfigure}[t]{0.24\textwidth}
    \centering
    \includegraphics[width=\linewidth]{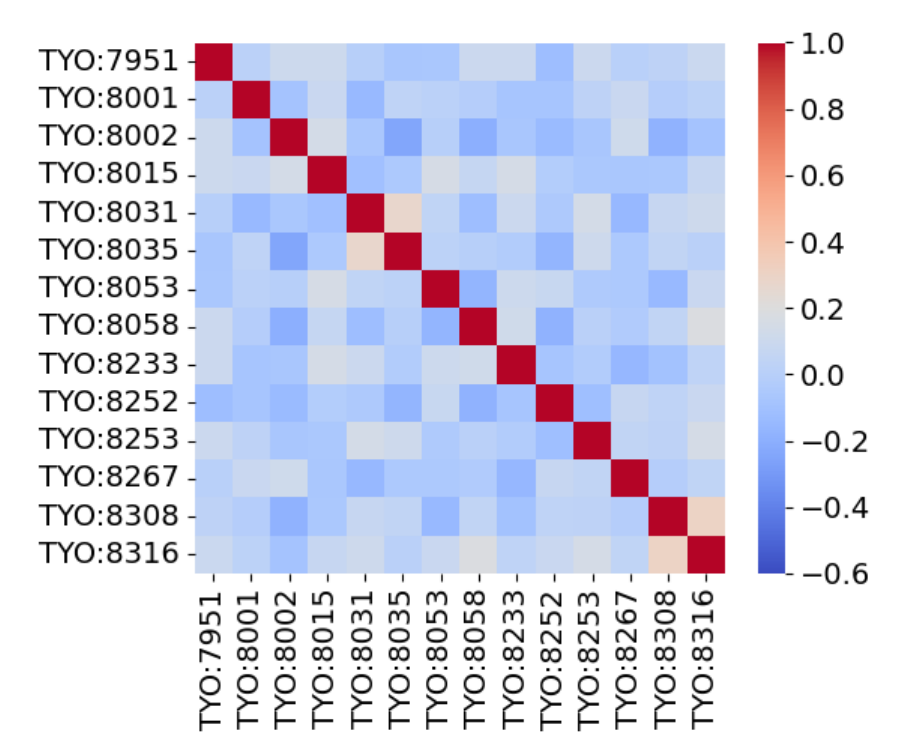}
    \caption{Shr. Whitening (Valid)}
  \end{subfigure}
  \begin{subfigure}[t]{0.24\textwidth}
    \centering
    \includegraphics[width=\linewidth]{Figures/Cor-Matrix/PCA+GGM-Valid.pdf}
    \caption{PCA+GGM (Valid)}
  \end{subfigure}
  \begin{subfigure}[t]{0.24\textwidth}
    \centering
    \includegraphics[width=\linewidth]{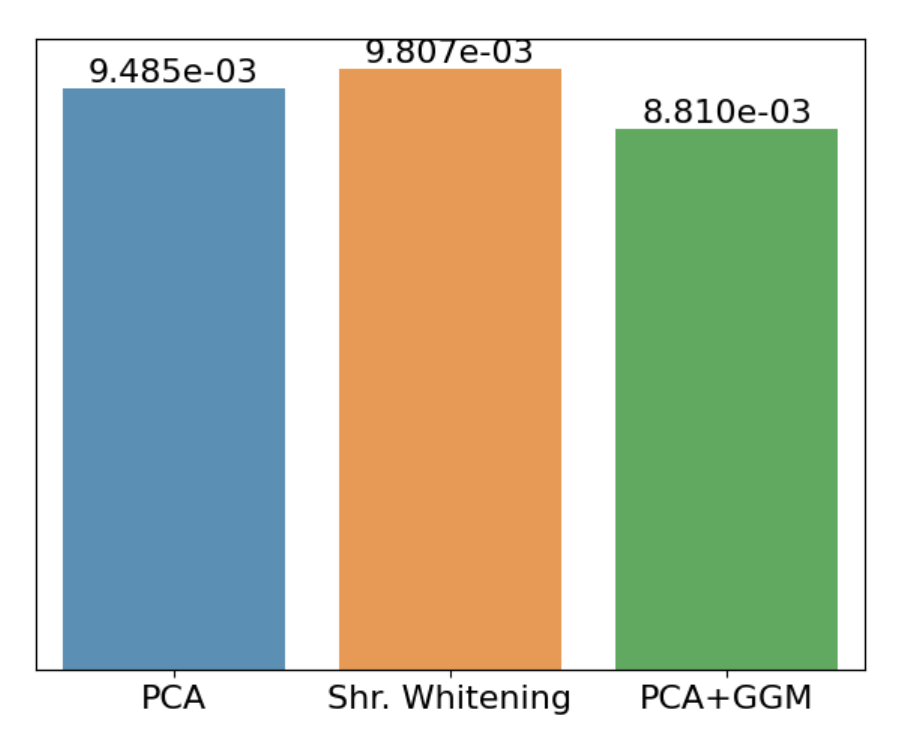}
    \caption{$\ell_2$ mean (Valid)}
  \end{subfigure}
  \caption{(a)–(c): Comparison of the correlation matrices estimated during the training period by each method. (e)–(g): Comparison of the correlation matrices estimated during the validation period by each method. The comparisons are shown for a subset of stocks, and the axes correspond to the respective ticker symbols. (d),(h): $\ell_2$-mean of the correlation matrices in the training and validation periods.}
  \label{Fig:Apx-Cor-Matrix}
\end{figure*}

\subsection{Additional Qualitative Results}
In this subsection, we conduct a qualitative assessment of correlation structures using the return data of TOPIX~ 500 constituents. 
In this experiment, we trained on the return series from \texttt{2020/01} through \texttt{2023/12} and used \texttt{2024/01} to \texttt{2024/12} as the test period. 
As in the previous experiments, we employed Shr.~Whitening and PCA as comparison baselines, and we used the $\ell_2$ mean as the evaluation metric.
\par
Figure~\ref{Fig:Apx-Cor-Matrix} presents a comparison of correlation matrices together with values of the evaluation metrics. 
In addition, Figures~\ref{Fig:Apx-Large-Raw-Cor-Matrix}--\ref{Fig:Apx-Large-PCA+GGM-Cor-Matrix} show larger correlation matrices during the validation period. 
In all cases, for clarity, we display only excerpts of the correlation matrices.
First, in the PCA results shown in Figure~\ref{Fig:Apx-Cor-Matrix}~(a), Figure~\ref{Fig:Apx-Cor-Matrix}~(e), and Figure~\ref{Fig:Apx-Large-PCA-Cor-Matrix}, residual correlations among specific sectors are observed. 
In the shrinkage whitening results shown in Figure~\ref{Fig:Apx-Cor-Matrix}~(b) and Figure~\ref{Fig:Apx-Cor-Matrix}~(f), no within-sector correlations are observed.
However, a comparison with Figure~\ref{Fig:Apx-Cor-Matrix}~(d), Figure~\ref{Fig:Apx-Cor-Matrix}~(h), and and Figure~\ref{Fig:Apx-Large-Shr-Cor-Matrix} indicates that the $\ell_2$ average of correlations is larger in the validation period, suggesting that whitening does not provide appropriate orthogonalization out of sample.
By contrast, the proposed method, shown in Figure~\ref{Fig:Apx-Cor-Matrix}~(c), Figure~\ref{Fig:Apx-Cor-Matrix}~(g), and Figure~\ref{Fig:Apx-Large-PCA+GGM-Cor-Matrix}, removes within-industry correlations while keeping the $\ell_2$ average in the validation period low, thereby confirming the effectiveness of the proposed approach.

\section{Evaluation Metrics}
\label{Apx:Evaluation-Metrics}
In this section, we provide a detailed description of the evaluation metrics used to assess trading performance.
Let $\mathbf{r}^{*}=[r_{1}^{*},r_{2}^{*},\ldots,r_{T^{*}}^{*}]\in\mathbb{R}^{T^{*}}$ denote the vector of returns obtained at each time step during the evaluation period of length $T^{*}$.
In this setting, the return is represented by the mean of the return series as $m^{*}=\frac{1}{T^{*}}\sum_{t=1}^{T^{*}}r^{*}_{t}$, and risk is represented by the standard deviation of the return series as $\sqrt{\frac{1}{T^{*}}\sum_{t=1}^{T^{*}}(r_{t}^{*}-m^{*})^2}$.
The Sharpe ratio is computed according to the following formula, with the risk-free rate denoted by $r_{f}$:
\begin{align}
    \mathrm{SR}=\frac{m^{*}-r_{f}}{s^{*}}.
    \label{Eq:Sharpe-Ratio}
\end{align}
Risk-free rates vary across markets. 
Because positions are closed at each time step, exposure to the risk-free rate is negligible; therefore, we set $r_f = 0$.
We evaluate downside risk using Conditional Value at Risk (CVaR), defined as the average loss in the left tail of the return distribution. 
Let $\mathcal{T}_{-}^{*}$ be the set of time indices corresponding to the worst $\alpha$ fraction of returns for the strategy (i.e., those at or below the empirical $\alpha$-quantile), and let $r_{t}^{*}$ denote the strategy’s return at time $t$. 
We report CVaR as a positive number:
\begin{align}
  \mathrm{CVaR} = \left|\frac{1}{|T_{-}^{*}|} \sum_{t \in \mathcal{T}_{-}^{*}} r_{t}^{*}\right|.
  \label{Eq:CVaR}
\end{align}
A smaller CVaR indicates milder losses in extreme scenarios, and thus a more stable trading strategy.
Maximum drawdown (MDD) is defined as the difference between the maximum and minimum of the cumulative return series and is computed as follows:
\begin{align}
    \mathrm{MDD}=\frac{\max_{t}(\mathbf{c}^{*})-\min_{t}(\mathbf{c}^{*})}{\max_{t}(\mathbf{c}^{*})},
\end{align}
where $\mathbf{c}^{*}\in\mathbb{R}^{T^{*}-1}$ is the cumulative returns of $\mathbf{r}^{*}$.

\begin{figure*}
    \centering
    \includegraphics[width=175mm]{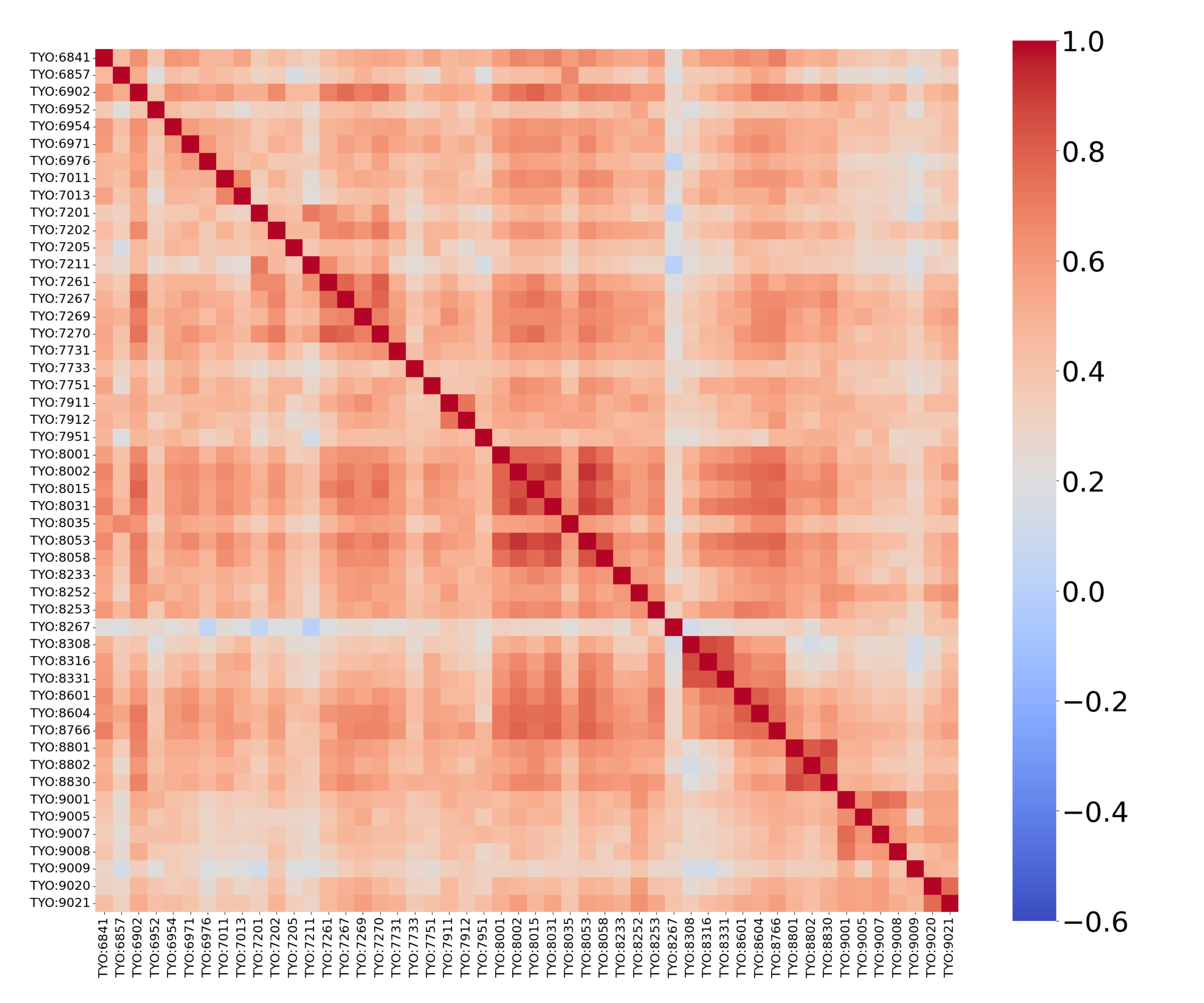}
    \caption{A correlation matrix of \textbf{raw returns} between assets of TOPIX~500 (Training period: \texttt{2020/01}--\texttt{2023/12}, Test period: \texttt{2024/01}--\texttt{2024/12}).}
    \label{Fig:Apx-Large-Raw-Cor-Matrix}
\end{figure*}

\begin{figure*}
    \centering
    \includegraphics[width=175mm]{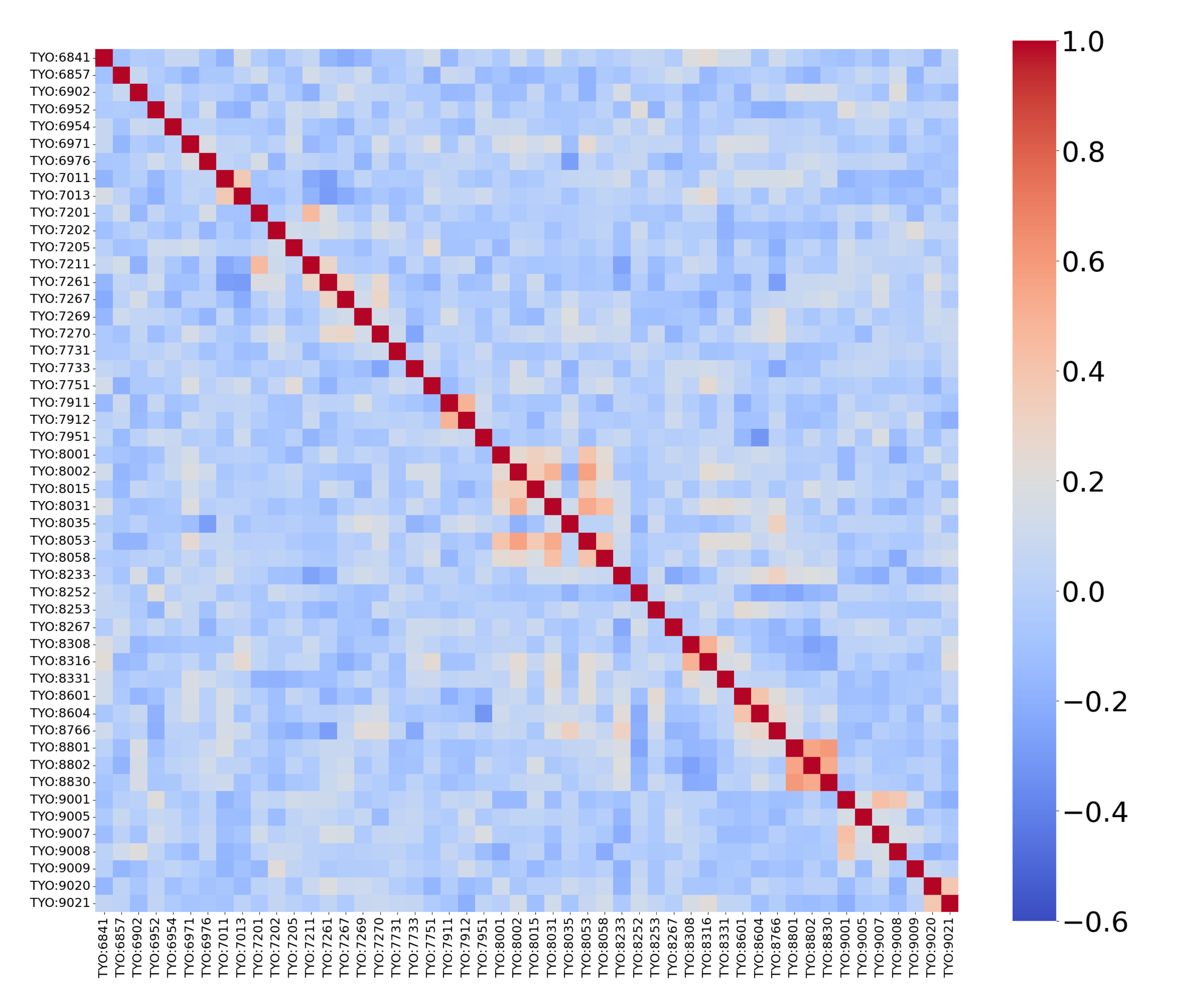}
    \caption{A correlation matrix of \textbf{PCA residual factors} between assets of TOPIX~500 (Training period: \texttt{2020/01}--\texttt{2023/12}, Test period: \texttt{2024/01}--\texttt{2024/12}).}
    \label{Fig:Apx-Large-PCA-Cor-Matrix}
\end{figure*}

\begin{figure*}
    \centering
    \includegraphics[width=175mm]{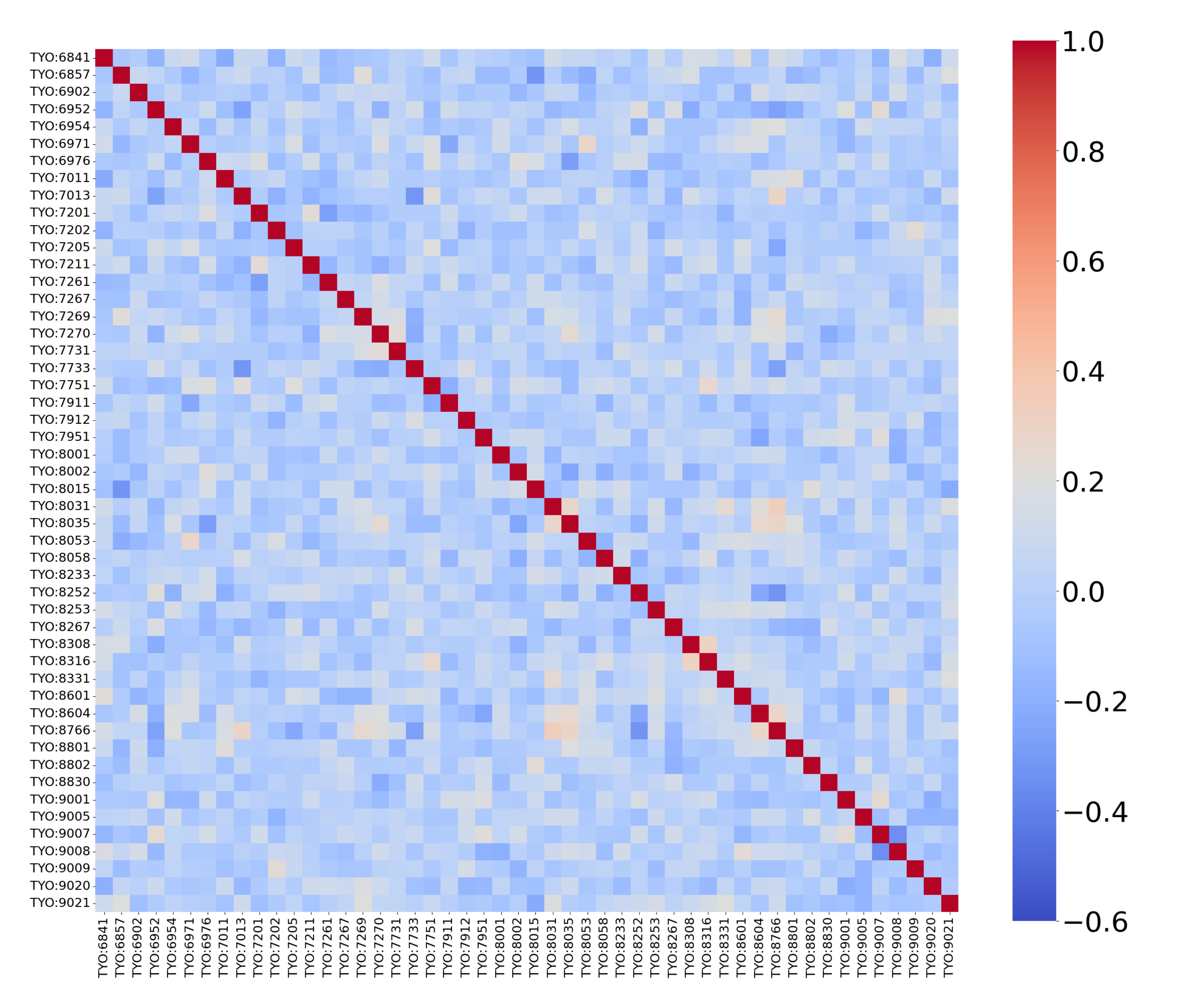}
    \caption{A correlation matrix of \textbf{Shr.~Whitening residual factors} between assets of TOPIX~500 (Training period: \texttt{2020/01}--\texttt{2023/12}, Test period: \texttt{2024/01}--\texttt{2024/12}).}
    \label{Fig:Apx-Large-Shr-Cor-Matrix}
\end{figure*}

\begin{figure*}
    \centering
    \includegraphics[width=175mm]{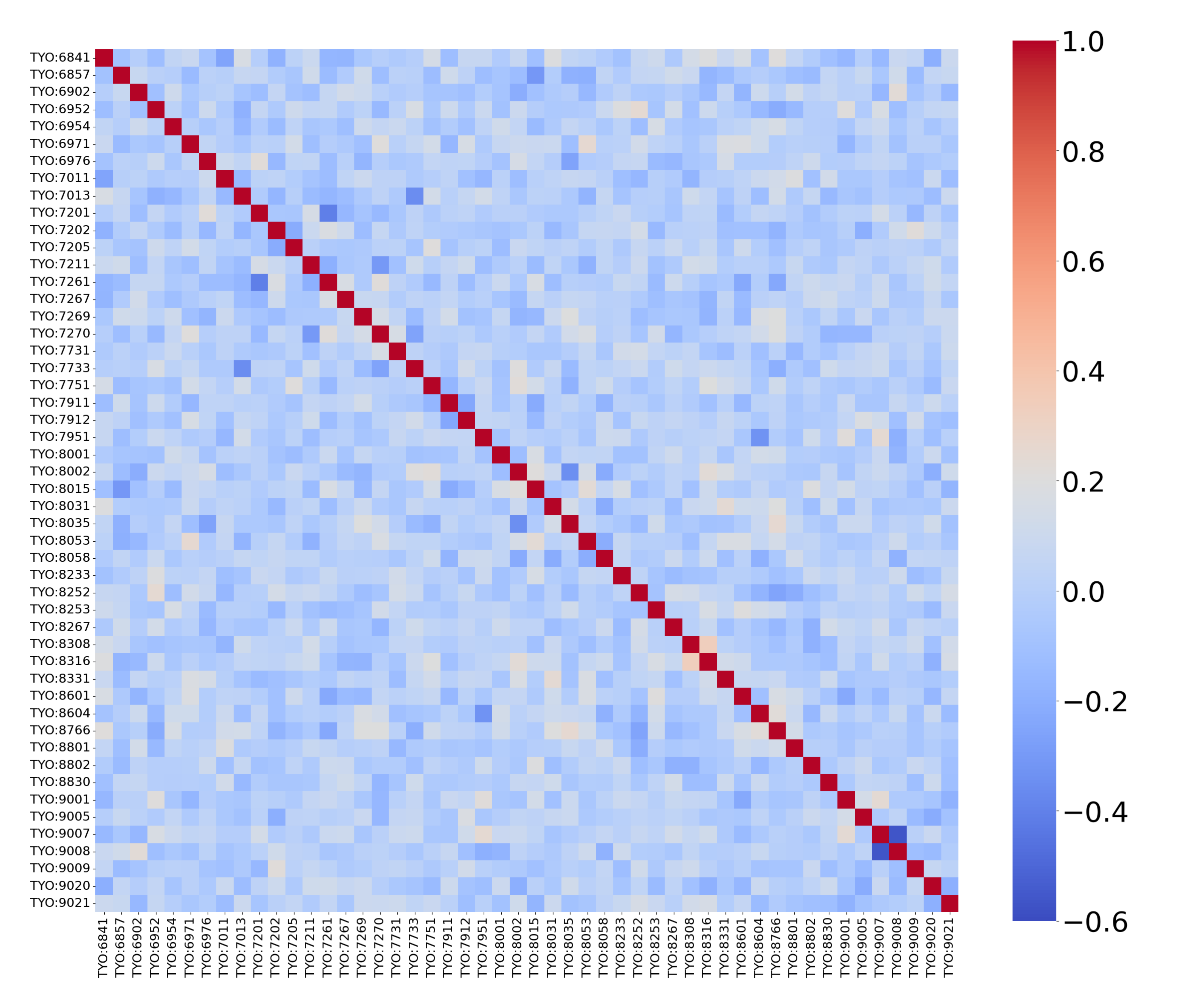}
    \caption{A correlation matrix of \textbf{PCA+GGM residual factors} between assets of TOPIX~500 (Training period: \texttt{2020/01}--\texttt{2023/12}, Test period: \texttt{2024/01}--\texttt{2024/12}).}
    \label{Fig:Apx-Large-PCA+GGM-Cor-Matrix}
\end{figure*}

\end{document}